\documentclass[journal,10pt]{IEEEtran}

\usepackage{algorithm,algorithmic}
\usepackage{comment}
   \usepackage[pdftex]{graphicx}

     \graphicspath{{Figures/}}

\usepackage{multicol, blindtext}
\usepackage{amsmath,amsthm}
\usepackage{enumerate}
\usepackage{amsfonts}
\title{Robust Kronecker Product PCA for Spatio-Temporal Covariance Estimation                      }

\newtheorem{theorem}{Theorem}[section]

\newtheorem{corollary}[theorem]{Corollary}

\begin{document}
\author{Kristjan~Greenewald,~\IEEEmembership{Student Member,~IEEE,} and~Alfred O.~Hero III,~\IEEEmembership{Fellow,~IEEE}

\thanks{K. Greenewald and A. Hero III are with the Department
of Electrical Engineering and Computer Science, University of Michigan, Ann Arbor,
MI, USA. This research was partially supported by grants from AFOSR FA8650-07-D-1220-0006 and ARO MURI W911NF-11-1-0391.}
}
  \maketitle
\begin{abstract}

Kronecker PCA involves the use of a space vs. time Kronecker product decomposition to estimate spatio-temporal covariances. 
In this work  the addition of a sparse correction factor is considered, which corresponds to a model of the covariance as a sum of Kronecker products of low (separation) rank and a sparse matrix. This sparse correction extends the diagonally corrected Kronecker PCA of \cite{greenewaldArxiv,greenewaldSSP2014} to allow for sparse unstructured ``outliers" anywhere in the covariance matrix, e.g. arising from variables or correlations that do not fit the Kronecker model well, or from sources such as sensor noise or sensor failure. We introduce a robust PCA-based algorithm to estimate the covariance under this model, extending the rearranged nuclear norm penalized LS Kronecker PCA approaches of \cite{greenewaldSSP2014,tsiliArxiv}. An extension to Toeplitz temporal factors is also provided, producing a parameter reduction for temporally stationary measurement modeling. High dimensional MSE performance bounds are given for these extensions.
Finally, the proposed extension of KronPCA is evaluated on both simulated and real data coming from yeast cell cycle experiments. This establishes the practical utility of robust Kronecker PCA in biological and other applications.

\end{abstract}


\section{Introduction}
\IEEEPARstart{I}{n} this paper, we develop a method for robust estimation of spatio-temporal covariances and apply it to multivariate time series modeling and parameter estimation. The covariance for spatio-temporal processes is manifested as a multiframe covariance, i.e. as the covariance not only between variables or features in a single frame (time point), but also between variables in a set of $p_t$ nearby frames. 
If each frame contains $p_s$ spatial variables, then the spatio-temporal covariance at time $t$ is described by a $p_t p_s$ by $p_t p_s$ matrix
\begin{equation}
\mathbf{\Sigma}_t = \mathrm{Cov}\left[\{\mathbf{I_n}\}_{n=t-p_t}^{t-1}\right],
\end{equation}
where $\mathbf{I}_n$ denotes the $p_s$ variables or features of interest in the $n$th frame. We make the standard piecewise stationarity assumption that $\mathbf{\Sigma}_t$ can be approximated as unchanging over each consecutive set of $p_t$ frames.

%

As $p_s p_t$ can be very large, even for moderately large $p_s$ and $p_t$ the number of degrees of freedom ($p_s p_t(p_s p_t+1)/2$) in the covariance matrix can greatly exceed the number $n$ of training samples available to estimate the covariance matrix. One way to handle this problem is to introduce structure and/or sparsity into the covariance matrix, thus reducing the number of parameters to be estimated. 

A natural non-sparse option is to introduce structure by modeling the covariance matrix $\mathbf{\Sigma}$ as the Kronecker product of two smaller symmetric positive definite matrices, i.e.
\begin{equation}
\label{KronApprox}
\mathbf{\Sigma} = \mathbf{A}\otimes \mathbf{B}.
\end{equation}
When the measurements are Gaussian with covariance of this form they are said to follow a matrix-normal distribution \cite{dutilleulFLIPFLOP,dawid1981some,tsiligkaridis2013convergence}. This model lends itself to coordinate decompositions \cite{tsiliArxiv}, such as the decomposition between space (variables) vs. time (frames) natural to spatio-temporal data \cite{tsiliArxiv,greenewaldArxiv}. In the spatio-temporal setting, the $p_s \times p_s$ $\mathbf{B}$ matrix is the ``spatial covariance" and the $p_t \times p_t$ matrix $\mathbf{A}$ is the ``time covariance," both identifiable up to a multiplicative constant. 

An extension to the representation \eqref{KronApprox}, discussed in \cite{tsiliArxiv}, approximates the covariance matrix using a sum of Kronecker product factors
\begin{equation}
\label{SumApprox}
\mathbf{\Sigma} = \sum\nolimits_{i=1}^{r} \mathbf{A}_i \otimes \mathbf{B}_i,
\end{equation}
where $r$ is the separation rank, $\mathbf{A}_i \in \mathbb{R}^{p_t \times p_t}$, and $\mathbf{B}_i \in \mathbb{R}^{p_s \times p_s}$. We call this the Kronecker PCA (KronPCA) covariance representation.

This model (with $r>1$) has been used in various applications, including video modeling and classification \cite{greenewaldArxiv,greenewaldSPIE2014}, network anomaly detection \cite{greenewaldSSP2014}, synthetic aperture radar, and MEG/EEG covariance modeling (see \cite{tsiliArxiv} for references). In \cite{loan1992approximation} it was shown that any covariance matrix can be represented in this form with sufficiently large $r$. 
This allows for more accurate approximation of the covariance when it is not in Kronecker product form but most of its energy can be accounted for by a few Kronecker components. An algorithm (Permuted Rank-penalized Least Squares (PRLS)) for fitting the model \eqref{SumApprox} to a measured sample covariance matrix was introduced in \cite{tsiliArxiv} and was shown to have strong high dimensional MSE performance guarantees. It should also be noted that, as contrasted to standard PCA,  KronPCA accounts specifically for spatio-temporal structure, often provides a full rank covariance, and requires significantly fewer components (Kronecker factors) for equivalent covariance approximation accuracy. Naturally, since it compresses covariance onto a more complex (Kronecker) basis than PCA's singular vector basis, the analysis of Kron-PCA estimation performance is more complicated.

The standard Kronecker PCA model does not naturally accommodate additive noise since the diagonal elements (variances) must conform to the Kronecker structure of the matrix. To address this issue, in \cite{greenewaldArxiv} we extended this KronPCA model, and the PRLS algorithm of \cite{tsiliArxiv}, by adding a structured diagonal matrix to \eqref{SumApprox}. This model is called Diagonally Loaded Kronecker PCA (DL-KronPCA) and, although it has an additional $p_sp_t$ parameters, it was shown that for fixed $r$ it performs significantly better for inverse covariance estimation in cases where there is additive measurement noise \cite{greenewaldArxiv}.

The DL-KronPCA model \cite{greenewaldArxiv} is the $r+1$-Kronecker model
\begin{equation}
\label{Eq:DiagKron}
\mathbf{{\Sigma}} = \left(\sum\nolimits_{i = 1}^{r}\mathbf{A}_i \otimes \mathbf{B}_i\right) + \mathbf{U} = \mathbf{\Theta} + \mathbf{U},
\end{equation}
where the diagonal matrix $\mathbf{U}$ is called the ``diagonal loading matrix." Following Pitsianis-VanLoan rearrangement of the square $p_tp_s\times p_tp_s$ matrix $\mathbf{\Sigma}$ to an equivalent rectangular $p_s^2 \times p_t^2$ matrix  \cite{tsiliArxiv,werner2008estimation}, this becomes an equivalent matrix approximation problem of finding a low rank plus diagonal approximation \cite{greenewaldArxiv,tsiliArxiv}. The DL-KronPCA estimation problem was posed in \cite{greenewaldSSP2014,greenewaldArxiv} as the rearranged nuclear norm penalized Frobenius norm optimization 
\begin{equation}
\min_{\mathbf{\Sigma}}\|\mathbf{\Sigma}-\hat{\mathbf{\Sigma}}_{SCM}\|_F^2+\lambda \|\mathcal{R}(\mathbf{\Theta})\|_*
\end{equation}
where the minimization is over $\mathbf{\Sigma}$ of the form \eqref{Eq:DiagKron}, $\mathcal{R}(\cdot)$ is the Pitsianis-VanLoan rearrangement operator defined in the next section, and $\| \cdot \|_*$ is the nuclear norm. A weighted least squares solution to this problem is given in \cite{greenewaldArxiv,greenewaldSSP2014}.

This paper extends DL-KronPCA to the case where $\mathbf{U}$ in \eqref{Eq:DiagKron} is a sparse loading matrix that is not necessarily diagonal. In other words, we model the covariance as the sum of a low separation rank matrix $\mathbf{\Theta}$ and a sparse matrix $\mathbf{\Gamma}$:
\begin{equation}
\label{Eq:SparseModel}
{\mathbf{\Sigma}} = \left(\sum\nolimits_{i = 1}^{r}\mathbf{A}_i \otimes \mathbf{B}_i\right) + \mathbf{\Gamma} = \mathbf{\Theta} + \mathbf{\Gamma}.
\end{equation}
DL-KronPCA is obviously a special case of this model. The motivation behind the extension \eqref{Eq:SparseModel} is that while the KronPCA models \eqref{SumApprox} and \eqref{Eq:DiagKron} may provide a good fit to most entries in $\mathbf{\Sigma}$, there are sometimes a few variables (or correlations) that cannot be well modeled using KronPCA, due to complex non-Kronecker structured covariance patterns, e.g. sparsely correlated additive noise, sensor failure, or corruption. Thus, inclusion of a sparse term in \eqref{Eq:SparseModel} allows for a better fit with lower separation rank $r$, thus reducing the overall number of model parameters. 
In addition, if the underlying distribution is heavy tailed, sparse outliers in the sample covariance will occur, which will corrupt Kronecker product estimates \eqref{SumApprox} and \eqref{Eq:DiagKron} that don't have the flexibility of absorbing them into a sparse term. 
This notion of adding a sparse correction term to a regularized covariance estimate is found in the Robust PCA literature, where it is used to allow for more robust and parsimonious approximation to data matrices \cite{chandrasekaran2009,chandrasekaran2010latent,candes2011robust,yang2013dirty}. Robust KronPCA differs from Robust PCA in that it replaces the outer product with the Kronecker product. KronPCA and PCA are useful for significantly different applications because the Kronecker product allows the decomposition of spatio-temporal processes into (full rank) spatio-temporally \emph{separable} components, whereas PCA decomposes them into deterministic basis functions with no explicit spatio-temporal structure \cite{tsiliArxiv,greenewaldArxiv,werner2008estimation}. Sparse correction strategies have also been applied in the regression setting where the sparsity is applied to the first moments instead of the second moments \cite{5540138,otazo2014low}.


The model \eqref{Eq:SparseModel} is called the Robust Kronecker PCA (Robust KronPCA) model, and we propose regularized least squares based estimation algorithms for fitting the model. In particular, we propose a singular value thresholding (SVT) approach using the rearranged nuclear norm. However, unlike in robust PCA, the sparsity is applied to the Kronecker decomposition instead of the singular value decomposition. 
We derive high dimensional consistency results for the SVT-based algorithm that specify the MSE tradeoff between covariance dimension and the number of samples. Following \cite{greenewaldSSP2014}, we also allow for the enforcement of a temporal block Toeplitz constraint, which corresponds to a temporally stationary covariance and results in a further reduction in the number of parameters when the process under consideration is temporally stationary and the time samples are uniformly spaced. We illustrate our proposed robust Kronecker PCA method using simulated data and a yeast cell cycle dataset.

The rest of the paper is organized as follows: in Section \ref{Sec:RobustKronPCA}, we introduce our Robust KronPCA model and introduce an algorithm for estimating covariances described by it. Section \ref{Sec:Consistency} provides high dimensional convergence theorems. Simulations and an application to cell cycle data are presented in Section \ref{Sec:Results}, and our conclusions are given in Section \ref{Sec:Conclusion}.

\section{Robust KronPCA}
\label{Sec:RobustKronPCA}

Let $\mathbf{X}$ be a $p_s \times p_t$ matrix with entries $\tilde{x}(m,t)$ denoting samples of a space-time random process defined over a $p_s$-grid of spatial samples $m\in \{1,\ldots, p_s\}$ and a $p_t$-grid of time samples $t\in \{1,\ldots, p_t\}$. Let $\mathbf{x}={\mathrm{ vec}}({\mathbf{X}})$ denote the $p_tp_s$ column vector obtained by lexicographical reordering. Define the $p_tp_s \times p_tp_s$ spatiotemporal covariance matrix $ \mathbf{\Sigma}=\mathrm{Cov}[\mathbf{x}]$. 


Consider the model \eqref{Eq:SparseModel} for the covariance as the sum of a low separation rank matrix $\mathbf{\Theta}$ and a sparse matrix $\mathbf{\Gamma}$:
\begin{equation}
\label{Eq:KronSp}
{\mathbf{\Sigma}} = \mathbf{\Theta} + \mathbf{\Gamma}.
\end{equation}

Define $\mathbf{M}(i,j)$ as the $i,j$th $p_s \times p_s$ subblock of $\mathbf{M}$, i.e., $\mathbf{M}(i,j) = [M]_{(i-1)p_s+1:ip_s, (j-1)p_s + 1:jp_s}$. The invertible Pitsianis-VanLoan rearrangement operator $\mathcal{R}(\cdot)$ maps $p_tp_s\times p_tp_s$ matrices to $p_t^2 \times p_s^2$ matrices and, as defined in \cite{tsiliArxiv,werner2008estimation} sets the $(i-1)p_t + j$th row of $\mathcal{R}(\mathbf{M})$ equal to $\mathrm{vec}(\mathbf{M}(i,j))^T$, i.e. 
\begin{align}
\mathcal{R}(\mathbf{M}) &= [\begin{array}{ccc} \mathbf{m}_1 & \dots & \mathbf{m}_{p_t^2}\end{array}]^T,\\\nonumber
\mathbf{m}_{(i-1)p_t+j} &= \mathrm{vec}(\mathbf{M}(i,j)), \quad i,j = 1,\dots,p_t.
\end{align}

After Pitsianis-VanLoan rearrangement the expression \eqref{Eq:KronSp} takes the form
\begin{equation}
\label{Eq:lowrank}
\mathcal{R}(\mathbf{\Sigma}) = \sum\nolimits_{i=1}^r \mathbf{a}_i \mathbf{b}_i^T+\mathbf{S} = \mathbf{L} + \mathbf{S},
\end{equation}
where $\mathbf{a}_i = \mathrm{vec}(\mathbf{A}_i)$ and $\mathbf{b}_i = \mathrm{vec}(\mathbf{B}_i)$. In the next section we solve this Robust Kronecker PCA problem (low rank + sparse + noise) using sparse approximation, involving a nuclear and 1-norm penalized Frobenius norm loss on the rearranged fitting errors. 

\subsection{Estimation}
\label{Sec:Estimation}

\label{S:Kron}


Similarly to the approach of \cite{greenewaldArxiv,tsiliArxiv}, we fit the model \eqref{Eq:KronSp} to the sample covariance matrix $\hat{\mathbf{ \Sigma}}_{SCM}=n^{-1}\sum_{i=1}^n (\mathbf{x}_i-\overline{\mathbf{x}}) (\mathbf{x}_i-\overline{\mathbf{x}})^T$, where $\overline{\mathbf{x}}$ is the sample mean and $n$ is the number of samples of the space time process $ \mathbf{X}$. The best fit matrices $ \mathbf{A}_i$, $ \mathbf{B}_i$ and $\mathbf{\Gamma}$ in \eqref{Eq:KronSp} are determined by minimizing the objective function
\begin{equation}
\label{Eq:OptProbNEw}
\min_{\hat{\mathbf{\Theta}},\hat{\mathbf{\Gamma}}}\| \hat{\mathbf{\Sigma}}_{SCM}-\hat{ \mathbf{\Theta}}-\hat{\mathbf{\Gamma}}\|_F^2+\lambda_\Theta\|\mathcal{R}(\hat{\mathbf{\Theta}})\|_* + \lambda_\Gamma \|\hat{\mathbf{\Gamma}}\|_1.
\end{equation}
We call the norm $\|\mathcal{R}(\mathbf{\Theta})\|_*$ the rearranged nuclear norm of $\mathbf{\Theta}$. The regularization parameters $\lambda_\Theta$ and $\lambda_\Gamma$ control the importance of separation rank deficiency and sparsity, respectively, where increasing either increases the amount of regularization. The objective function \eqref{Eq:OptProbNEw} is equivalent to the rearranged objective function
\begin{equation}
\label{Eq:SparseOpt}
\min_{\hat{\mathbf{L}},\hat{\mathbf{S}}}\| \mathbf{R}-\hat{ \mathbf{L}}-\hat{\mathbf{S}}\|_F^2+\lambda_\Theta\|\hat{\mathbf{L}}\|_* + \lambda_\Gamma \|\hat{\mathbf{S}}\|_1,
\end{equation}
with $ \mathbf{R}=\mathcal{R}( \hat{\mathbf{\Sigma}}_{SCM})$. The objective function is minimized over all $p_t^2 \times p_s^2$ matrices $\hat{ \mathbf{R}} = \hat{\mathbf{L}} + \hat{\mathbf{S}}$. The solutions $\hat{\mathbf{L}}$ and $\hat{\mathbf{S}}$ correspond to estimates of $\mathcal{R}(\mathbf{\Theta})$ and $\mathcal{R}(\mathbf{\Gamma})$ respectively. 
As shown in \cite{greenewaldArxiv}, the left and right singular vectors of $\hat{\mathbf{L}}$ correspond to the (normalized) vectorized $\mathbf{A}_i$ and $\mathbf{B}_i$ respectively, as in \eqref{Eq:lowrank}.

This nuclear norm penalized low rank matrix approximation is a well-studied optimization problem \cite{mazumder2010spectral}, where it is shown to be strictly convex. Several fast solution methods are available, including the iterative SVD-based proximal gradient method on which Algorithm \ref{alg:SVTNon} is based \cite{mooreSSP2014}.
If the sparse correction $\mathbf{\Gamma}$ is omitted, equivalent to setting $\lambda_\Gamma = \infty$, the resulting optimization problem can be solved directly using the SVD \cite{tsiliArxiv}. The minimizers $\hat{\mathbf{L}}$, $\hat{\mathbf{S}}$ of \eqref{Eq:SparseOpt} are transformed to the covariance estimate by the simple operation
\begin{align}
\label{Eq:KronPCA}
\hat{\mathbf{\Sigma}} = \mathcal R^{-1}\left( \hat{\mathbf{L}}+\hat{\mathbf{S}} \right),
\end{align}
where $\mathcal{R}^{-1}(\cdot)$ is the inverse of the permutation operator $\mathcal{R}(\cdot)$. 
As the objective function in Equation \eqref{Eq:SparseOpt} is strictly convex and is equivalent to the Robust PCA objective function of \cite{mooreSSP2014}, Algorithm \ref{alg:SVTNon} converges to a unique global minimizer. 

Algorithm \ref{alg:SVTNon} is an appropriate modification of the iterative algorithm described in \cite{mooreSSP2014}. It consists of alternating between two simple steps: 1) soft thresholding of the singular values of the difference between the overall estimate and the estimate of the sparse part ($\mathbf{SVT}_\lambda(\cdot)$), and 2) soft thresholding of the entries of the difference between the overall estimate and the estimate of the low rank part ($\mathbf{soft}_\lambda(\cdot)$). The soft singular value thresholding operator is defined as  
\begin{align}
\mathbf{SVT}_\lambda(\mathbf{M}) = \mathbf{U}\left(\mathrm{diag}(\sigma_1,\dots,\sigma_{\min(m_1,m_2)}) - \lambda \mathbf{I}\right)_+ \mathbf{V}^T,
\end{align}
where $\mathbf{U}\mathrm{diag}(\sigma_1,\dots,\sigma_{\min(m_1,m_2)})\mathbf{V}^T$ is the singular value decomposition of $\mathbf{M} \in \mathbb{R}^{m_1\times m_2}$ and $(\cdot)_+ = \max (\cdot, 0)$. The entrywise soft thresholding operator is given by
\begin{align}
[\mathbf{soft}_\lambda(\mathbf{M})]_{ij} &= \mathrm{sign}({M}_{ij})(|{M}_{ij}| - \lambda)_+.
\end{align}


\begin{algorithm}[H]
\caption{Proximal Gradient Robust KronPCA}
\label{alg:SVTNon}
\begin{algorithmic}[1]

\STATE $\mathbf{R} =\mathcal{R}(\hat{\mathbf{\Sigma}}_{SCM})$
\STATE Initialize $\mathbf{M},\mathbf{S},\mathbf{L}$, choose step sizes $\tau_k$.
\WHILE{$\mathcal R^{-1}\left(  {\mathbf{L}}+{\mathbf{S}} \right)$ not converged}
\STATE $\mathbf{L}^k = \mathbf{SVT}_{\tau_k\lambda_\Theta'}(\mathbf{M}^{k-1}-\mathbf{S}^{k-1})$
\STATE $\mathbf{S}^k = \mathbf{soft}_{\tau_k\lambda_\Gamma'}(\mathbf{M}^{k-1}-\mathbf{L}^{k-1})$
\STATE $\mathbf{M}^k = \mathbf{L}^k + \mathbf{S}^k - \tau_k(\mathbf{L}^k + \mathbf{S}^k-\mathbf{R})$
\ENDWHILE
\STATE $\hat{\mathbf{\Sigma}} = \mathcal R^{-1}\left(  {\mathbf{L}}+{\mathbf{S}} \right)$
\RETURN $\hat{\mathbf{\Sigma}}$
\end{algorithmic}
\end{algorithm}

\begin{algorithm}[H]
\caption{Proximal Gradient Robust Toeplitz KronPCA}
\label{alg:SVT}
\begin{algorithmic}[1]

\STATE $\tilde{\mathbf{R}} = \mathbf{P}\mathcal{R}(\hat{\mathbf{\Sigma}}_{SCM})$
\STATE Initialize $\mathbf{M},\mathbf{S},\mathbf{L}$, choose step sizes $\tau_k$.
\WHILE{$\mathcal R^{-1}\left( \mathbf{P}^T( {\mathbf{L}}+{\mathbf{S}}) \right)$ not converged}
\STATE $\mathbf{L}^k = \mathbf{SVT}_{\tau_k\lambda_\Theta'}(\mathbf{M}^{k-1}-\mathbf{S}^{k-1})$
\FOR{ $j \in \mathcal{I}$}
\STATE $\mathbf{S}^k_{j+p_t} = \mathbf{soft}_{\tau_k\lambda_\Gamma'c_j}(\mathbf{M}^{k-1}_{j+p_t}-\mathbf{L}^{k-1}_{j+p_t})$
\ENDFOR
\STATE $\mathbf{M}^k = \mathbf{L}^k + \mathbf{S}^k - \tau_k(\mathbf{L}^k + \mathbf{S}^k-\tilde{\mathbf{R}})$
\ENDWHILE
\STATE $\hat{\mathbf{\Sigma}} = \mathcal R^{-1}\left(\mathbf{P}^{T} \left( {\mathbf{L}}+{\mathbf{S}} \right)\right)$
\RETURN $\hat{\mathbf{\Sigma}}$
\end{algorithmic}
\end{algorithm}

\subsection{Block Toeplitz Structured Covariance}
Here we extend Algorithm \ref{alg:SVTNon} to incorporate a block Toeplitz constraint. Block Toeplitz constraints are relevant to stationary processes arising in signal and image processing. For simplicity we consider the case that the covariance is block Toeplitz with respect to time, however, extensions to the cases of Toeplitz spatial structure and having Toeplitz structure simultaneously in both time and space are straightforward. The objective function \eqref{Eq:SparseOpt}, is to be solved with a constraint that both $\hat{\mathbf{\Theta}}$ and $\hat{\mathbf{\Gamma}}$ are temporally block Toeplitz.

For low separation rank component $\mathbf{\Theta} = \sum_{i=1}^r \mathbf{A}_i \otimes \mathbf{B}_i$, the block Toeplitz constraint is equivalent to a Toeplitz constraint on the temporal factors $\mathbf{A}_i$. The Toeplitz constraint on $\mathbf{A}_i$ is equivalent to \cite{kamm2000optimal,pitsianis1997kronecker}
\begin{equation}
\label{Eq:Toeplitz}
[\mathbf{a}_i]_{k} = v^{(i)}_{j + p_t},\: \forall k \in \mathcal{K}(j), j =-p_t+1,\dots, p_t-1,
\end{equation}
for some vector $\mathbf{v}^{(i)}$ where $\mathbf{a}_i = \mathrm{vec}(\mathbf{A}_i)$ and
\begin{equation}
\mathcal{K}(j) = \{k : \: (k-1)p_t + k + j \in [ -p_t+1,\: p_t-1]\}.
\end{equation}

It can be shown, after some algebra, that the optimization problem \eqref{Eq:SparseOpt} constrained to block Toeplitz covariances is equivalent to an unconstrained penalized least squares problem involving $\mathbf{v}^{(i)}$ instead of $\mathbf{a}_i$. Specifically, following the techniques of \cite{kamm2000optimal,pitsianis1997kronecker,greenewaldSSP2014} with the addition of the 1-norm penalty, the constrained optimization problem \eqref{Eq:SparseOpt} can be shown to be equivalent to the following unconstrained optimization problem:
\begin{align}
\label{Eq:OptProbToep}
\min_{{\tilde{\mathbf{L}}},\tilde{\mathbf{S}}} \|\tilde{\mathbf{R}} - {\tilde{\mathbf{L}}}-{\tilde{\mathbf{S}}} \|_F^2 + &\lambda_{\Theta}' \| \tilde{\mathbf{L}} \|_*+ \lambda_\Gamma' \sum_{j \in \mathcal{I}} c_j \| \tilde{\mathbf{S}}_{j+p_t} \|_1,
\end{align}
where $\tilde{\mathbf{S}}_j$ denotes the $j$th row of $\tilde{\mathbf{S}}$, $\tilde{\mathbf{L}} = \mathbf{P}\hat{\mathbf{L}}$, $\tilde{\mathbf{S}} = \mathbf{P}\mathbf{S}$, and $\tilde{\mathbf{R}}= \mathbf{P}\mathbf{R}$. The summation indices $\mathcal{I}$, the 1-norm weighting constants $c_j$, and the $(2p_t-1) \times p_t^2$ matrix $\mathbf{P}$ are defined as
\begin{align}
\label{Eq:ToepCoef}
\mathcal{I} =& \{-p_t+1,\dots,p_t-1\}\\\nonumber
c_j =& 1/\sqrt{p_t - |j|}\\\nonumber
P_{p_t+j,i} =& \left\{\begin{array}{ll}\frac{1}{\sqrt{p_t - |j|}}& i \in \mathcal{K}(j)\\ 0&o.w.\end{array}\right. 
\end{align}
where the last line holds for all $j = -p_t + 1,\dots, p_t-1, i =1,\dots, p_s^2$. 
Note that this imposition of Toeplitz structure also results in a significant reduction in computational cost primarily due to a reduction in the size of the matrix in the singular value thresholding step \cite{greenewaldSSP2014}. 
The block Toeplitz estimate is given by
\begin{align}
\label{Eq:KronPCAToep}
\hat{\mathbf{\Sigma}} = \mathcal R^{-1}\left(\mathbf{P}^{T}\left(\tilde{\mathbf{L}}+\tilde{\mathbf{S}}\right) \right),
\end{align}
where $\tilde{\mathbf{L}}$, $\tilde{\mathbf{S}}$ are the minimizers of \eqref{Eq:OptProbToep}.
Similarly to the non-Toeplitz case, the block Toeplitz estimate can be computed using Algorithm \ref{alg:SVT}, which is the appropriate modification of Algorithm \ref{alg:SVTNon}. 
As the objective function in Equation \eqref{Eq:SparseOpt} is strictly convex and is equivalent to the Robust PCA objective function of \cite{mooreSSP2014}, Algorithm \ref{alg:SVT} converges to a unique global minimizer. 


The non-Toeplitz and Toeplitz objective functions \eqref{Eq:SparseOpt} and \eqref{Eq:OptProbToep}, respectively, are both invariant with respect to replacing $\mathbf{\Theta}$ with $\mathbf{\Theta}^T$ and $\mathbf{\Gamma}$ with $\mathbf{\Gamma}^T$ because $\mathbf{\Sigma}_{SCM}$ is symmetric. Furthermore, $\|(\mathbf{M} + \mathbf{M}^T)/2\| \leq \frac{1}{2}(\|\mathbf{M}\|+ \|\mathbf{M}^T\|) = \|\mathbf{M}\|$ (for both the weighted 1-norm and nuclear norm) by the triangle inequality. Hence the symmetric $(\mathbf{\Sigma} + \mathbf{\Sigma}^T)/2 = (\mathbf{\Theta} + \mathbf{\Theta}^T)/2 + (\mathbf{\Gamma} + \mathbf{\Gamma}^T)/2$ will always result in at least as low an objective function value as would $\mathbf{\Sigma}$. 
By the uniqueness of the global optimum, 
the Robust KronPCA covariance estimates $\hat{\mathbf{\Theta}}$ and $\hat{\mathbf{\Gamma}}$ are therefore symmetric for both the Toeplitz and non-Toeplitz cases.

\section{High Dimensional Consistency}
\label{Sec:Consistency}
In this section, we impose the additional assumption that the training data $\mathbf{x}_i$ is Gaussian with true covariance given by 
\begin{equation}
\mathbf{\Sigma}_0 = \mathbf{\Theta}_0 + \mathbf{\Gamma}_0,
\end{equation}
where $\mathbf{\Theta}_0$ is the low separation rank covariance of interest and $\mathbf{\Gamma}_0$ is sparse. 

A norm $\mathcal{R}_{k}(\cdot)$ is said to be \emph{decomposable} with respect to subspaces ($\mathcal{M}_k, \bar{\mathcal{M}}_k$) if \cite{yang2013dirty}
\begin{equation}
\mathcal{R}_k (u + v) = \mathcal{R}_k(u) + \mathcal{R}_k (v), \quad \forall u\in \mathcal{M}_k, v \in \bar{\mathcal{M}}_k^\perp.
\end{equation}

We define subspace pairs ${\mathcal{M}}_\mathbf{Q},\bar{\mathcal{M}}_\mathbf{Q}$ \cite{yang2013dirty} with respect to which the rearranged nuclear ($\mathbf{Q} = \mathbf{\Theta}$) and 1-norms ($\mathbf{Q} = \mathbf{\Gamma}$) are respectively decomposable \cite{yang2013dirty}. These are associated with the set of either low separation rank ($\mathbf{Q} = \mathbf{\Theta}$) or sparse ($\mathbf{Q} = \mathbf{\Gamma}$) matrices. For the sparse case, let $S$ be the set of indices on which $\mathrm{vec}\{\mathbf{\Gamma}\}$ is nonzero. Then $\mathcal{M}_\Gamma = \bar{\mathcal{M}}_\Gamma$ is the subspace of vectors in $\mathbb{R}^{p_t^2p_s^2}$ that have support contained in $S$, and $\bar{\mathcal{M}}_\Gamma^\perp$ is the subspace of vectors orthogonal to $\bar{\mathcal{M}}_\Gamma^\perp$, i.e. the subspace of vectors with support contained in $S^c$. 

For the rearranged nuclear norm, note that by \cite{loan1992approximation} any $pq \times pq$ matrix $\mathbf{\Theta}$ can be decomposed as
\begin{align}
\label{Eq:Decomp}
  \mathbf{\Theta} = \sum_{i=1}^{\min(p_t^2,p_s^2)} \sigma_i \mathbf{A}_{\Theta}^{(i)}\otimes \mathbf{B}_{\Theta}^{(i)}
\end{align}
where for all $i$, $\|\mathbf{A}_{\Theta}^{(i)}\|_F = \|\mathbf{B}_{\Theta}^{(i)}\|_F = 1$, $\sigma_i \geq 0$ and nonincreasing,
the $p_t\times p_t$ $\{\mathbf{A}_{\Theta}^{(i)}\}_i$ are all linearly independent, and the
$p_s \times p_s$ $\{\mathbf{B}_{\Theta}^{(i)}\}_i$ are all linearly independent.
It is easy to show that this decomposition can be computed by extracting and rearranging the singular value decomposition of
$\mathcal{R}(\mathbf{\Theta})$ \cite{loan1992approximation,tsiliArxiv,werner2008estimation} and thus the $\sigma_i$ are uniquely determined. Let $r$ be such that $\sigma_i = 0$
for all $i > r$. Define the matrices 
\begin{align}
{\mathbf{U}}_{A} &=[\mathrm{vec}(\mathbf{A}_{\Theta}^{(1)}), \dots, \mathrm{vec}(\mathbf{A}_{\Theta}^{(r)})], \nonumber\\\nonumber
{\mathbf{U}}_{B}&=[\mathrm{vec}(\mathbf{B}_{\Theta}^{(1)}), \dots,
\mathrm{vec}(\mathbf{B}_{\Theta}^{(r)})].
\end{align}
Then we define a pair of subspaces with respect to which the nuclear norm is decomposable as 
\begin{align}
\mathcal{M}_\Theta =& \mathrm{range}({\mathbf{U}_A}\otimes {\mathbf{U}_B}),\\\nonumber
\bar{\mathcal{M}}^\perp_\Theta =& \mathrm{range}({\mathbf{U}_A^\perp}\otimes {\mathbf{U}_B^\perp}).
\end{align} 
It can be shown that these subspaces are uniquely determined by $\mathbf{\Theta}$.

Consider the covariance estimator that results from solving the optimization problem in Equation \eqref{Eq:SparseOpt}.
As is typical in Robust PCA, an incoherence assumption is required to ensure that $\mathbf{\Theta}$ and $\mathbf{\Gamma}$ are distinguisable.
Our incoherence assumption is as follows:
\begin{align}
\label{Eq:INCOH}
&\max\left\{\sigma_{max}\left(\mathcal{P}_{\bar{\mathcal{M}}_{\mathbf{\Theta}}}\mathcal{P}_{\bar{\mathcal{M}}_{\mathbf{\Gamma}}}\right),\sigma_{max}\left(\mathcal{P}_{\bar{\mathcal{M}}_{\mathbf{\Theta}}^\bot}\mathcal{P}_{\bar{\mathcal{M}}_{\mathbf{\Gamma}}}\right),\right.\\\nonumber
&\left.\sigma_{max}\left(\mathcal{P}_{\bar{\mathcal{M}}_{\mathbf{\Theta}}}\mathcal{P}_{\bar{\mathcal{M}}_{\mathbf{\Gamma}}^\bot}\right),\sigma_{max}\left(\mathcal{P}_{\bar{\mathcal{M}}_{\mathbf{\Theta}}^\bot}\mathcal{P}_{\bar{\mathcal{M}}_{\mathbf{\Gamma}}^\bot}\right)\right\} \leq \frac{16}{\Lambda^2}
\end{align}
where
\begin{equation}
\Lambda = 2 + \max\left\{\frac{3\beta\sqrt{2r}}{\lambda\sqrt{s}},\frac{3\lambda\sqrt{s}}{\beta\sqrt{2r}}\right\},
\end{equation}
$\mathcal{P}_{\bar{\mathcal{M}}_\mathbf{Q}}$ is the matrix corresponding to the projection operator that projects onto the subspace $\bar{\mathcal{M}}_\mathbf{Q}$ and $\sigma_{max}$ denotes the maximum singular value.

By way of interpretation, note that the maximum singular value of the product of projection matrices measures the ``angle" between the subspaces. Hence, the incoherence condition is imposing that the
subspaces in which $\mathbf{\Theta}$ and $\mathbf{\Gamma}$ live be sufficiently ``orthogonal'' to each other i.e., ``incoherent.'' This ensures identifiability in the sense that a portion of $\mathbf{\Gamma}$ (a portion of $\mathbf{\Theta}$) cannot be well approximated by adding a small number of additional terms to the Kronecker factors of $\mathbf{\Theta}$ ($\mathbf{\Gamma}$). Thus $\mathbf{\Theta}$ cannot be sparse and $\mathbf{\Gamma}$ cannot have low separation rank. In \cite{yang2013dirty} it was noted that this incoherency condition is significantly weaker than other typically imposed approximate orthogonality conditions.

Suppose that in the robust KronPCA model the $n$ training samples are multivariate Gaussian distributed and IID, that
$\mathbf{L}_0 = \mathcal{R}(\mathbf{\Theta}_0)$ is at most rank $r$, and that $\mathbf{S}_0 = \mathcal{R}(\mathbf{\Theta}_0)$ has $s$ nonzero entries \eqref{Eq:lowrank}.
Choose the regularization parameters to be
\begin{equation}
\label{Eq:beta}
\lambda_\Theta = k \|\mathbf{\Sigma}_0\| \max(\alpha^2,\alpha),\quad \lambda_\Gamma = 32\rho(\mathbf{\Sigma}_0)\sqrt{\frac{\log p_tp_s}{n}},
\end{equation}
where $\rho(\mathbf{\Sigma}) = \max_j \Sigma_{jj}$ and $k$ is smaller than an increasing function of $t_0$ given in the proof (\eqref{Eq:k}). We define $\alpha$ below. 

Given these assumptions, we have the following bound on the estimation error (defining $M = \max(p_t,p_s,n)$).
\begin{theorem}[Robust KronPCA]
\label{Thm:HDC}
Let $\mathbf{L}_0 = \mathcal{R}(\mathbf{\Theta}_0)$. Assume that the incoherence assumption \eqref{Eq:INCOH} and the assumptions in the previous paragraph hold, and the regularization parameters $\lambda_\Theta$ and $\lambda_\Gamma$ are given by Equation \eqref{Eq:beta} with $\alpha = \sqrt{t_0(p_t^2 + p_s^2 + \log M)/n}$ for any chosen $t_0>1$. Then the Frobenius norm error of the solution to the optimization problem \eqref{Eq:SparseOpt} is bounded as:
\begin{align}
\label{Eq:27}
\|\hat{\mathbf{L}} - &\mathbf{L}_0\|_F \leq\\\nonumber&6\max \left\{k\|\mathbf{\Sigma_0}\| \sqrt{r}\max(\alpha^2,\alpha), 32 \rho(\mathbf{\Sigma_0})\sqrt{\frac{s\log p_tp_s}{n}} \right\}
\end{align}
with probability at least $1-c\exp(-c_0 \log p_tp_s)$, where $c,c_0$ are constants, and $c_0$ is dependent on $t_0$ but is bounded from above by an absolute constant. 
\end{theorem}
The proof of this theorem can be found in Appendix \ref{App:A}. Note that in practice the regularization parameters will be chosen via a method such as cross validation, so given that the parameters in \eqref{Eq:beta} depend on $\mathbf{\Sigma}_0$, the specific values in \eqref{Eq:beta} are less important than how they scale with $p_t,p_s,n$.

Next, we derive a similar bound for the Toeplitz Robust KronPCA estimator.

It is easy to show that a decomposition of $\mathbf{\Theta}$ of the form \eqref{Eq:Decomp} exists where all the $\mathbf{A}_i$ are Toeplitz. Hence the definitions of the relevant subspaces ($\bar{\mathcal{M}}_{\mathbf{\Gamma}}$, $\bar{\mathcal{M}}_{\mathbf{\Theta}}$, ${\mathcal{M}}_{\mathbf{\Gamma}}$, ${\mathcal{M}}_{\mathbf{\Theta}}$) are of the same form as for the non Toeplitz case. In the Gaussian robust Toeplitz KronPCA model \eqref{Eq:OptProbToep}, further suppose $\tilde{\mathbf{L}}_0 = \mathbf{P}\mathbf{L}_0$ is at most rank $r$ and $\tilde{\mathbf{S}}_0 = \mathbf{P}\mathbf{S}_0$ has at most $s$ nonzero entries.

\begin{theorem}[Toeplitz Robust KronPCA]
\label{Thm:HDCToep}
Assume that the assumptions of Theorem III.1 hold and that $\mathbf{P L}_0$ is at most rank $r$ and that $\mathbf{P S}_0$ has at most $s$ non-zero entries. Let the regularization parameters $\lambda_\Theta$ and $\lambda_\Gamma$ be as in \eqref{Eq:beta} with $\alpha = \sqrt{t_0(2p_t + p_s^2 + \log M)/n}$ for any $t_0 > 1$. Then the Frobenius norm error of the solution to the Toeplitz Robust KronPCA optimization problem \eqref{Eq:SparseOpt} with coefficients given in \eqref{Eq:ToepCoef} is bounded as:
\begin{align}
\label{Eq:28}
\|\hat{\mathbf{L}} - &\mathbf{L}_0\|_F \leq\\\nonumber&6\max \left\{k\|\mathbf{\Sigma}_0\| \sqrt{r}\max(\alpha^2,\alpha), 32 \rho(\mathbf{\Sigma}_0)\sqrt{\frac{s\log p_tp_s}{n}} \right\}
\end{align}
with probability at least $1-c\exp(-c_0 \log p_tp_s)$, where $c,c_0$ are constants, and $c_0$ is dependent on $t_0$ but is bounded from above by an absolute constant.
\end{theorem}
The proof of this theorem is given in Appendix \ref{App:A}.

Comparing the right hand sides of \eqref{Eq:27} and \eqref{Eq:28} in Theorems III.1 and III.2 we see that, as expected, prior knowledge of Toeplitz structure reduces the Frobenius norm error of the estimator $\hat{\mathbf{L}}$ from $O(p_t^2)$ to $O(p_t)$.

\section{Results}
\label{Sec:Results}
\subsection{Simulations}
In this section we evaluate the performance of the proposed robust Kronecker PCA algorithms by conducting mean squared covariance estimation error simulations. For the first simulation, we consider a covariance that is a sum of 3 Kronecker products ($p_t = 10,\: p_s = 50$), with each term being a Kronecker product of two autoregressive (AR) covariances. AR processes with AR parameter $a$ have covariances $\mathbf{\Psi}$ given by
\begin{equation}
\psi_{ij} = c a^{|i-j|}.
\end{equation}
For the $p \times p$ temporal factors $\mathbf{A}_i$, we use AR parameters $[0.5, 0.8, 0.05]$ and for the $q \times q$ spatial factors $\mathbf{B}_i$ we use $[0.95, 0.35, 0.999]$. The Kronecker terms are scaled by the constants $[1,0.5,0.3]$. These values were chosen to create a complex covariance with 3 strong Kronecker terms with widely differing structure. The result is shown in Figure \ref{Fig:Cov}. We ran the experiments below for 100 cases with randomized AR parameters and in every instance Robust KronPCA dominated standard KronPCA and the sample covariance estimators to an extent qualitatively identical to the case shown below.

To create a covariance matrix following the non-Toeplitz ``KronPCA plus sparse" model, we create a new ``corrupted" covariance by taking the 3 term AR covariance and deleting a random set of row/column pairs, adding a diagonal term, and sparsely adding high correlations (whose magnitude depends on the distance to the diagonal) at random locations. Figure \ref{Fig:CovCor} shows the resulting corrupted covariance. To create a ``corrupted" block Toeplitz ``KronPCA plus sparse" covariance, a diagonal term and block Toeplitz sparse correlations were added to the AR covariance in the same manner as in the non-Toeplitz case. 

Figures \ref{Fig:MSEToep} and \ref{Fig:MSENT} show results for estimating the Toeplitz corrupted covariance, and non-Toeplitz corrupted covariance respectively, using Algorithms 1 and 2.  For both simulations, the average MSE of the covariance estimate is computed for a range of Gaussian training sample sizes. Average 3-steps ahead prediction MSE loss using the learned covariance to form the predictor coefficients ($\hat{\mathbf{\Sigma}}_{yx}\hat{\mathbf{\Sigma}}_{xx}^{\dag}$) is shown in Figure \ref{Fig:MSEP}. MSE loss is the prediction MSE using the estimated covariance minus $E[(y-E[y|x])^2]$, which is the prediction MSE achieved using an infinite number of training samples. The regularization parameter values shown are those used at $n=10^5$ samples, the values for lower sample sizes are set proportionally using the $n$-dependent formulas given by \eqref{Eq:beta} and Theorem \ref{Thm:HDC}. The chosen values of the regularization parameters are those that achieved best average performance in the appropriate region. Note the significant gains achieved using the proposed regularization, and the effectiveness of using the regularization parameter formulas derived in the theorems. In addition, note that separation rank penalization alone does not maintain the same degree of performance improvement over the unregularized (SCM) estimate in the high sample regime, whereas the full Robust KronPCA method maintains a consistent advantage (as predicted by the Theorems \ref{Thm:HDC} and \ref{Thm:HDCToep}).

%
%
%



\begin{figure}[htb]
\centering
\includegraphics[width=3.6in]{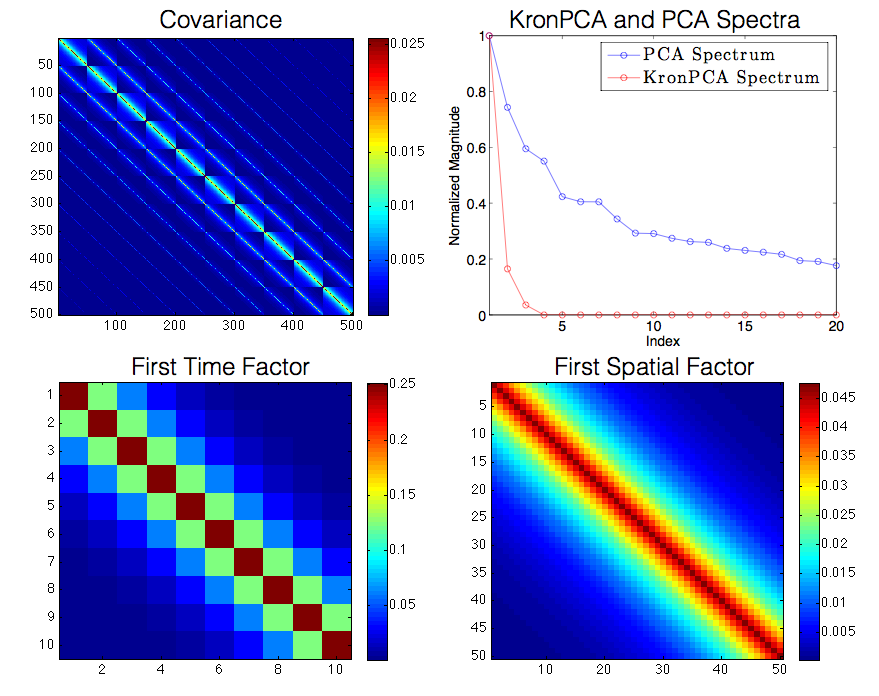}
\caption{Uncorrupted covariance used for the MSE simulations. Subfigures (clockwise starting from upper left): $r=3$ KronPCA covariance with AR factors $\mathbf{A}_i$, $\mathbf{B}_i$ \eqref{SumApprox}; its KronPCA and PCA spectra; and the first spatial and temporal factors of the original covariance.  }
\label{Fig:Cov}
\end{figure}
\begin{figure}[htb]
\centering
\includegraphics[width=3.6in]{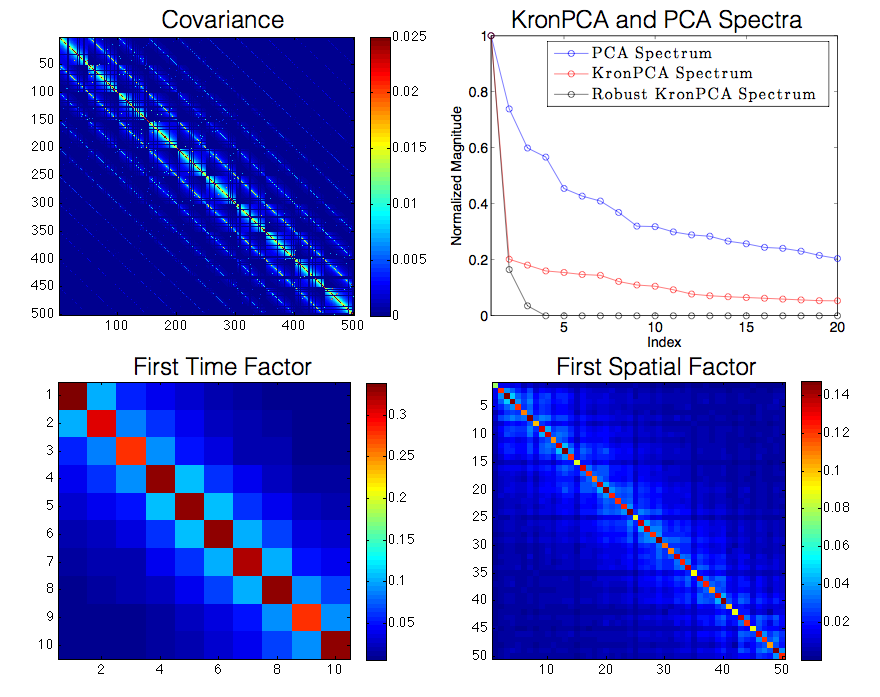}
\caption{Corrupted version of the $r=3$ KronPCA covariance (Figure \ref{Fig:Cov}), used to test the robustness of the proposed Robust KronPCA algorithm. Subfigures (Clockwise from upper left): covariance with sparse corruptions \eqref{Eq:SparseModel}; its KronPCA, Robust KronPCA, and PCA spectra; and the (non-robust) estimates of the first spatial and temporal factors of the corrupted covariance. Note that the corruption spreads the KronPCA spectrum and the significant corruption of the first Kronecker factor in the non-robust KronPCA estimate.}
\label{Fig:CovCor}
\end{figure}

%

\begin{figure}[htb]
\centering
\includegraphics[width=3.6in]{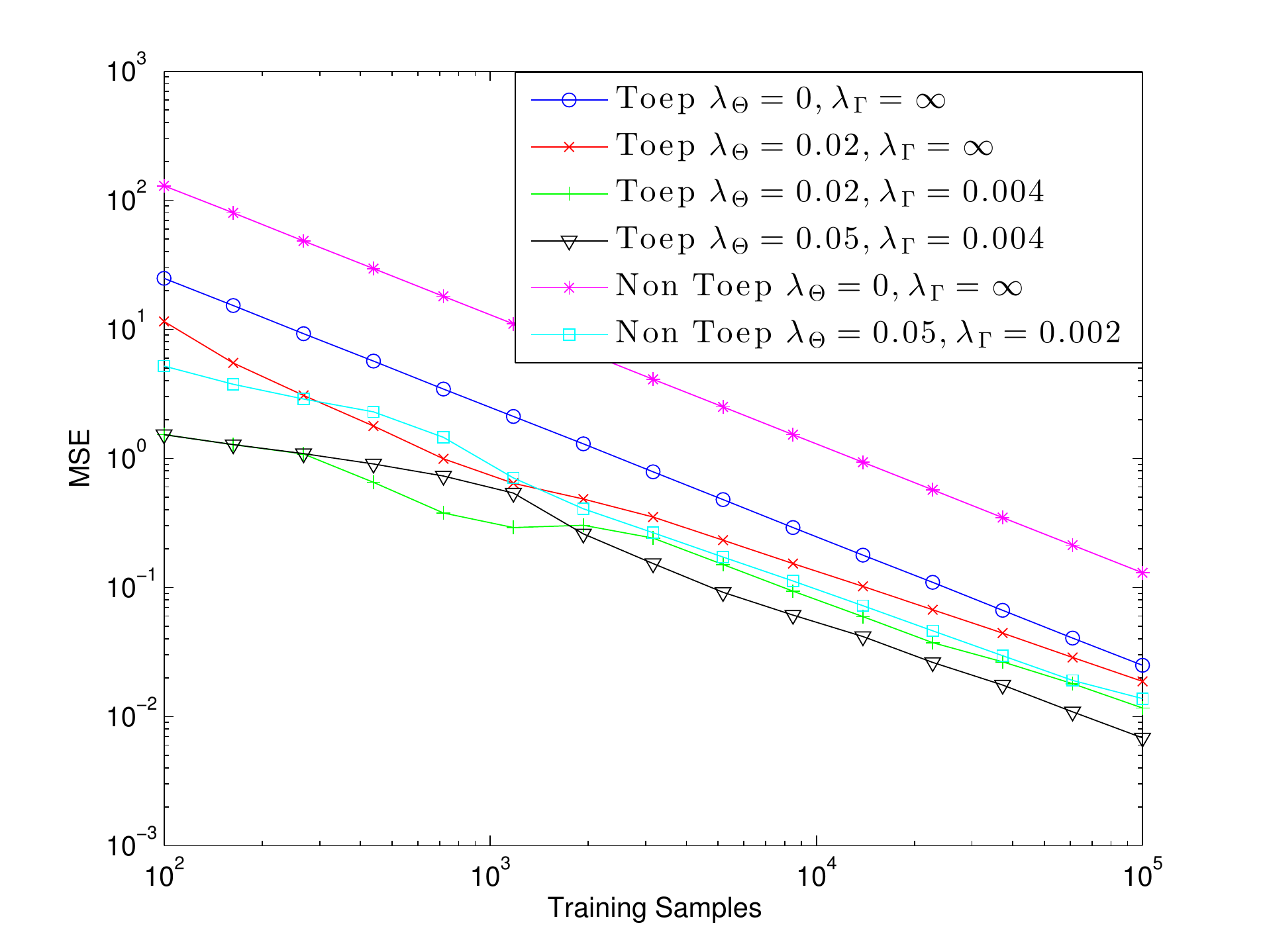}
\caption{MSE plots for both Toeplitz Robust KronPCA (Toep) and non Toeplitz Robust KronPCA (Non Toep) estimation of the Toeplitz corrupted covariance, as a function of the number of training samples. Note the advantages of using each of Toeplitz structure, separation rank penalization, and sparsity regularization, as proposed in this paper. The regularization parameter values shown are those used at $n=10^5$ samples, the values for lower sample sizes are set proportionally using the $n$-dependent formulas given by \eqref{Eq:beta} and Theorems \ref{Thm:HDC} and \ref{Thm:HDCToep}. }
\label{Fig:MSEToep}
\end{figure}

\begin{figure}[htb]
\centering
\includegraphics[width=3.6in]{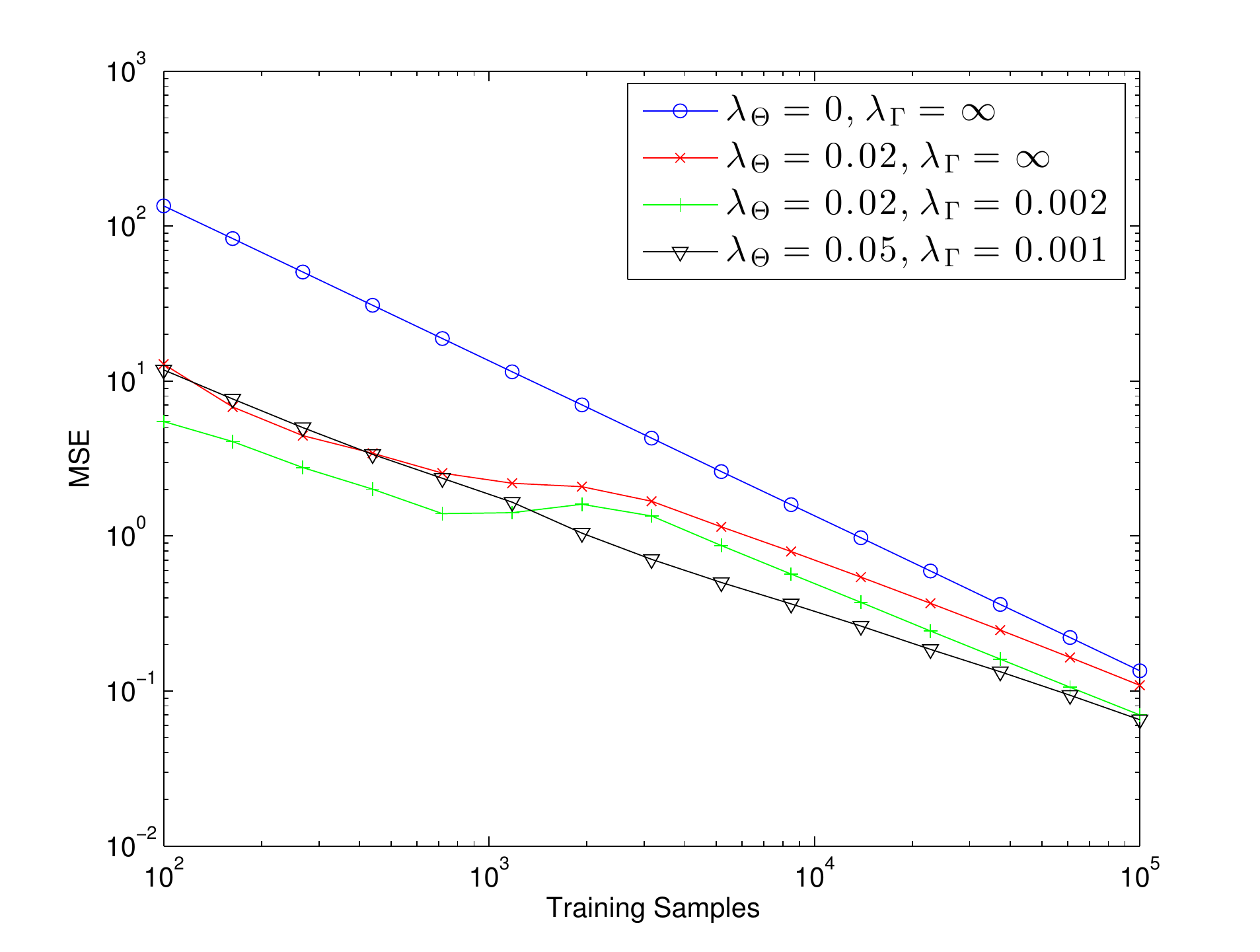}
\caption{MSE plots for non-Toeplitz Robust KronPCA estimation of the corrupted covariance, as a function of the number of training samples. Note the advantages of using both separation rank and sparsity regularization. The regularization parameter values shown are those used at $n=10^5$ samples, the values for lower sample sizes are set proportionally using the $n$-dependent formulas given by \eqref{Eq:beta} and Theorem \ref{Thm:HDC}. }
\label{Fig:MSENT}
\end{figure}

\begin{figure}[htb]
\centering
\includegraphics[width=3.6in]{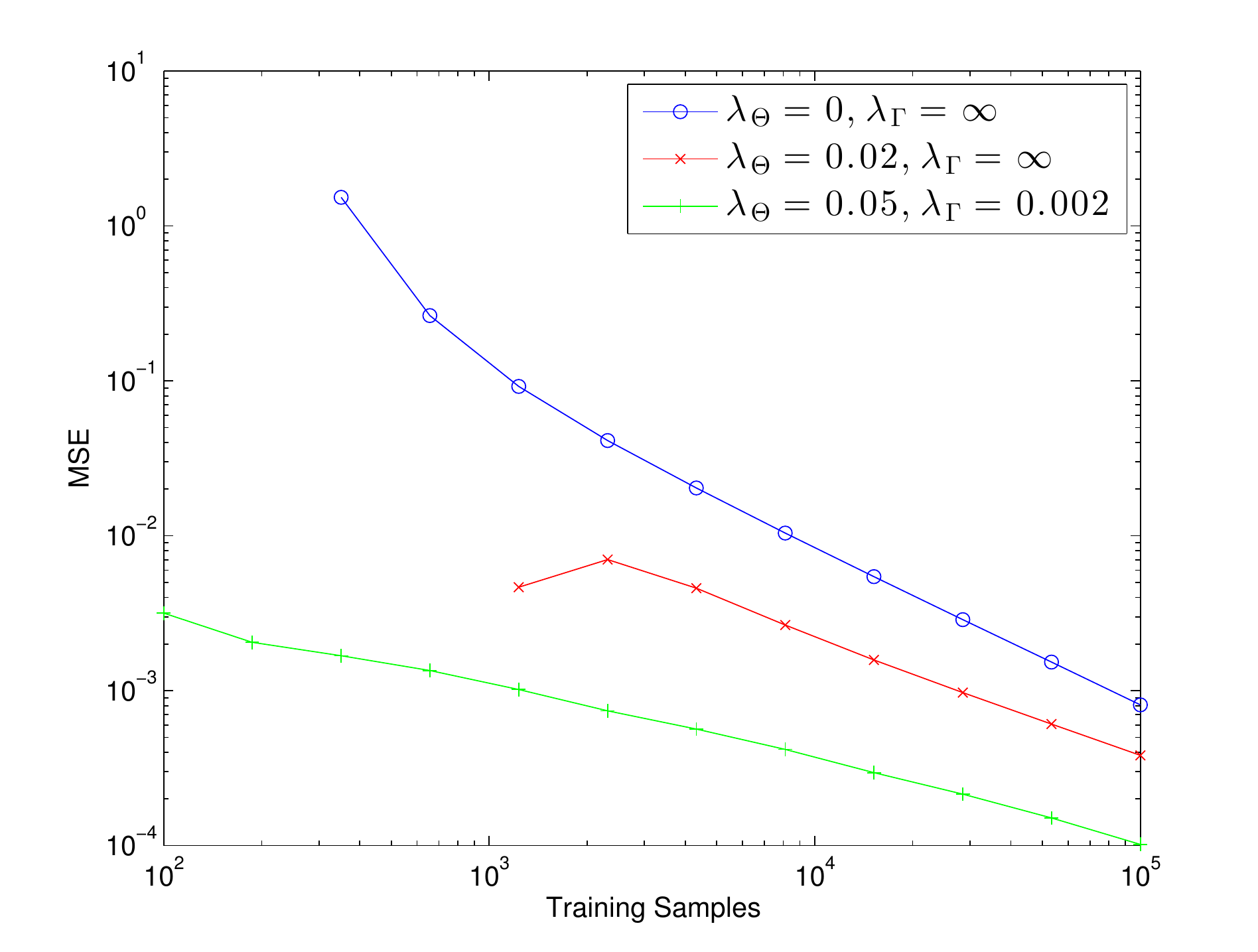}
\caption{3-ahead prediction MSE loss plots using the OLS predictor with corrupted covariance estimated by non-Toeplitz Robust KronPCA. Note the advantages of using both separation rank and sparsity regularization. The regularization parameter values shown are those used at $n=10^5$ samples, the values for lower sample sizes are set proportionally using the $n$-dependent formulas given by \eqref{Eq:beta} and Theorem \ref{Thm:HDC}. The sample covariance ($\lambda_\Theta = 0, \: \lambda_\Gamma = \infty$) and standard KronPCA ($\lambda_\Theta = 0.02, \: \lambda_\Gamma = \infty$) curves are cut short due to aberrant behavior in the low sample regime.}
\label{Fig:MSEP}
\end{figure}
\subsection{Cell Cycle Modeling}

As a real data application, we consider the yeast (S. cerevisiae) metabolic cell cycle dataset used in \cite{deckard2013design}. The dataset consists of 9335 gene probes sampled approximately every 24 minutes for a total of 36 time points, and about 3 complete cell cycles \cite{deckard2013design}.

In \cite{deckard2013design}, it was found that the expression levels of many genes exhibit periodic behavior in the dataset due to the periodic cell cycle. Our goal is to establish that periodicity can also be detected in the temporal component of the Kronecker spatio-temporal correlation model for the dataset. Here the spatial index is the label of the gene probe. We use $p_t=36$ so only one spatio-temporal training sample is available. Due to their high dimensionality, the spatial factor estimates have very low accuracy, but the first few temporal $\mathbf{A}_i$ factors ($36\times 36$) can be effectively estimated (bootstrapping using random sets of 20\% genes achieved less than 3\% RMS variation) due to the large number of spatial variables. We learn the spatiotemporal covariance (both space and time factors) using Robust KronPCA and then analyze the estimated time factors ($\mathbf{A}_i$) to discover periodicity. 
This allows us to consider the overall periodicity of the gene set, taking into account relationships between the genes, as contrasted to the univariate analysis as in \cite{deckard2013design}. The sparse correction to the covariance allows for the partial or complete removal of genes and correlations that are outliers in the sense that their temporal behavior differs from the temporal behavior of the majority of the genes.

Figure \ref{Fig:CorsHist} shows the quantiles of the empirical distribution of the entries of the sample covariance versus those of the normal distribution. The extremely heavy tails motivate the use of a sparse correction term as opposed to the purely quadratic approach of standard KronPCA. Plots of the first row of each temporal factor estimate are shown in Figure \ref{Fig:GenePeriod}. The first three factors are shown when the entire 9335 gene dataset is used to create the sample covariance. Note that 3 cycles of strong temporal periodicity are discovered, which matches our knowledge that approximately 3 complete cell cycles are contained in the sequence. Figure \ref{Fig:GenePeriod2} displays the estimates of the first temporal factor when only a random 500 gene subset is used to compute the sample covariance. Note that the standard KronPCA estimate has much higher variability than the proposed robust KronPCA estimate, masking the presence of periodicity in the temporal factor. This is likely due to the heavy tailed nature of the distribution, and to the fact that robust KronPCA is better able to handle outliers via the sparse correction of the low Kronecker-rank component.   


\begin{figure}[htb]
\centering
\includegraphics[width=3.6in]{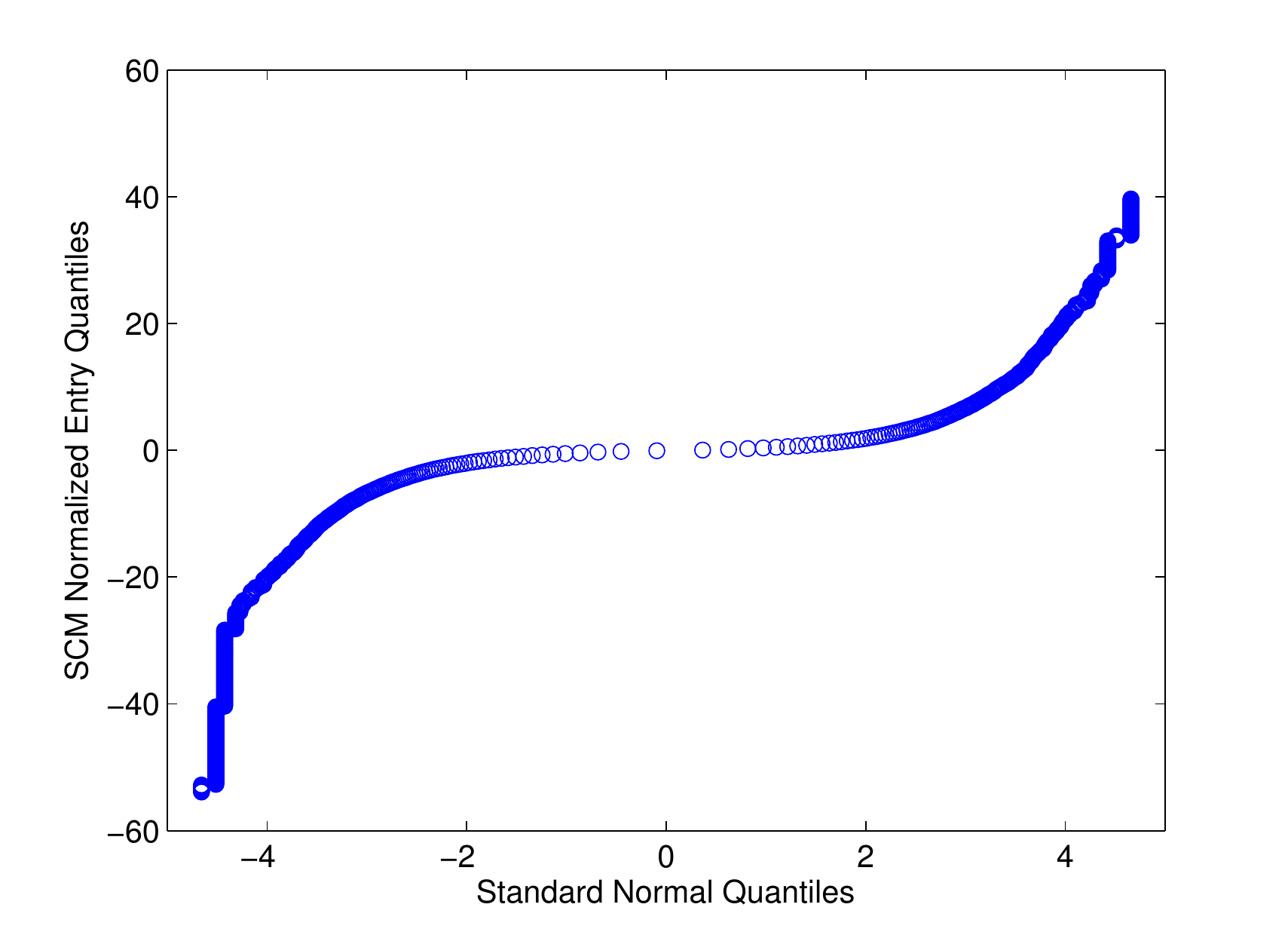}
\caption{Plot of quantiles of the empirical distribution of the sample covariance entries versus those of the normal distribution (QQ). Note the very heavy tails, suggesting that an 2-norm based approach will break down relative to the Robust KronPCA approach allowing for sparse corrections.}
\label{Fig:CorsHist}
\end{figure}

\begin{figure}[htb]
\begin{centering}
\includegraphics[width=3in]{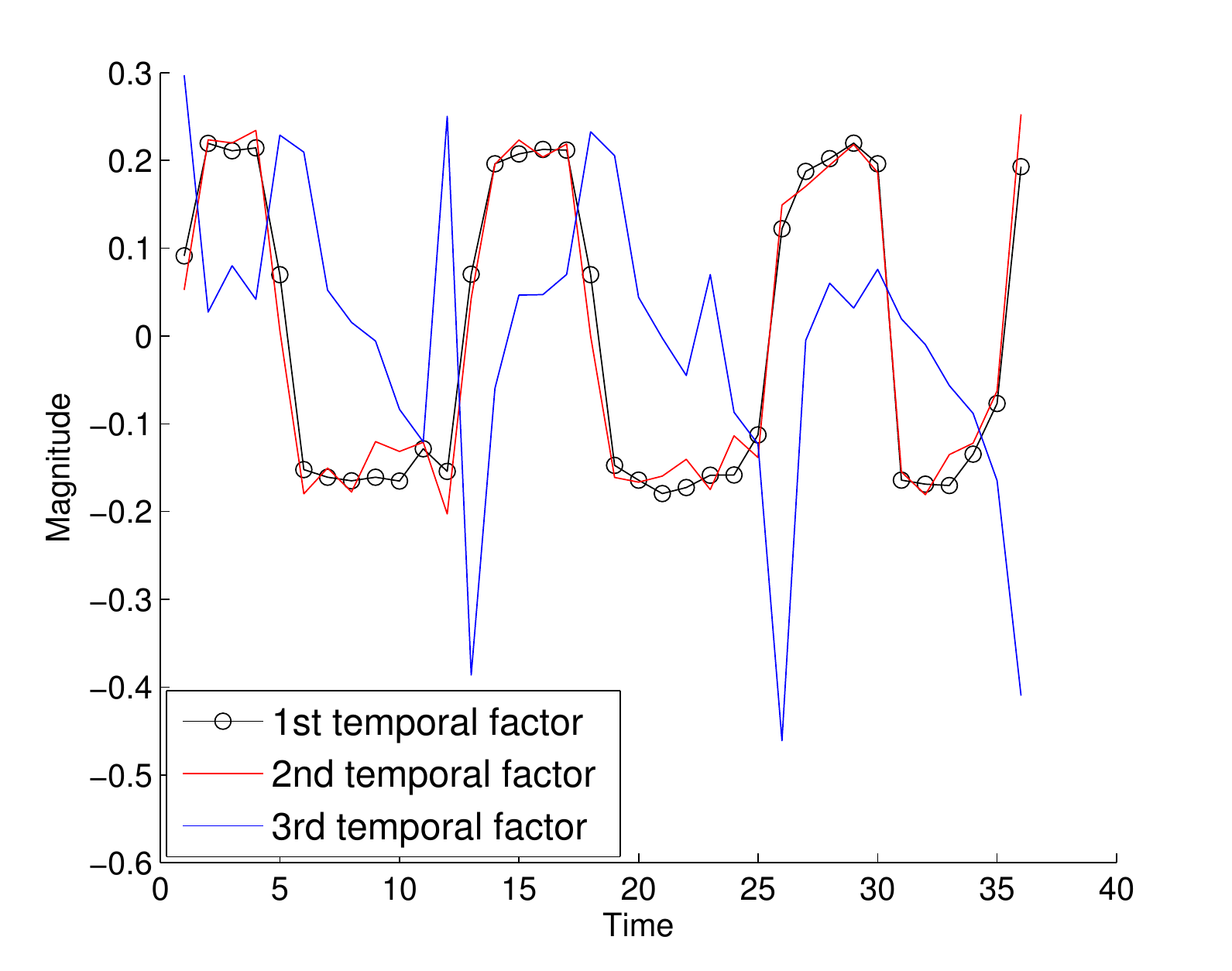}\\
\end{centering}
\caption{Plots of the temporal covariance factors estimated from the entire cell cycle dataset. Shown are the first rows of the first three temporal factors (excluding the first entry). Note the strong periodicity of the first two. }
\label{Fig:GenePeriod}
\end{figure}
\begin{figure}[htb]
\begin{centering}
\includegraphics[width=3.0in]{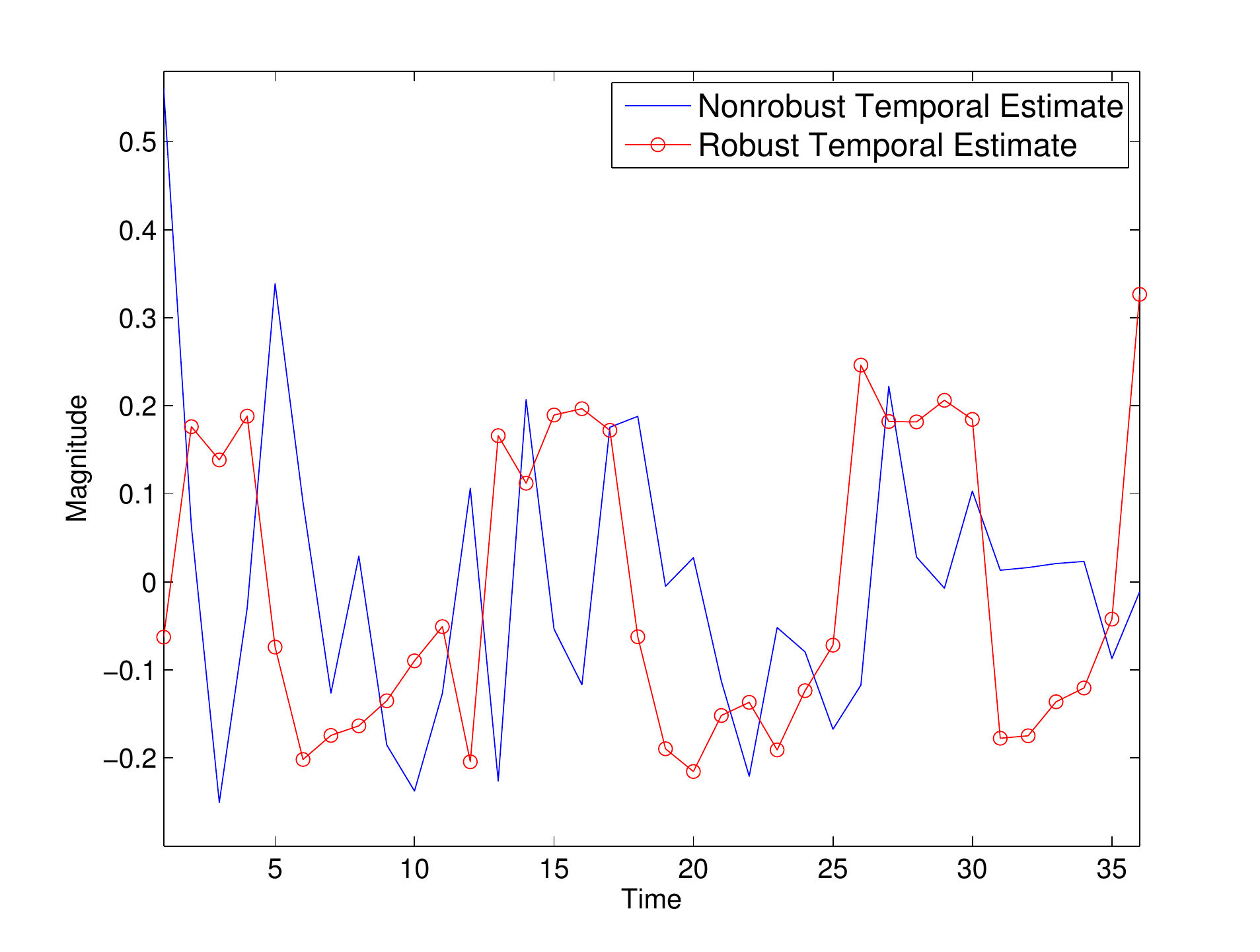}\\
\end{centering}
\caption{Plots of the first temporal covariance factors (excluding the first entry) estimated from the highly subsampled (spatially) cell cycle dataset using robust and standard KronPCA. Note the ability of Robust KronPCA to discover the correct periodicity.}
\label{Fig:GenePeriod2}
\end{figure}

\section{Conclusion}
\label{Sec:Conclusion}
This paper proposed a new robust method for performing Kronecker PCA of sparsely corrupted spatiotemporal covariances in the low sample regime. The method consists of a combination of KronPCA, a sparse correction term, and a temporally block Toeplitz constraint.
%
%
To estimate the covariance under these models, a robust PCA based algorithm was proposed. The algorithm is based on nuclear norm penalization of a Frobenius norm objective to encourage low separation rank, and high dimensional performance guarantees were derived for the proposed algorithm. 
Finally, simulations and experiments with yeast cell cycle data were performed that demonstrate advantages of our methods, both relative to the sample covariance and relative to the standard (non-robust) KronPCA.
\begin{appendices}
\section{Derivation of High Dimensional Consistency}
\label{App:A}
In this section we first prove Theorem \ref{Thm:HDC} and then prove Theorem \ref{Thm:HDCToep}, which are the bounds for the non-Toeplitz and Toeplitz cases respectively.

A general theorem for decomposable regularization of this type was proven in \cite{yang2013dirty}. In \cite{yang2013dirty} the theorem was applied to Robust PCA directly on the sample covariance, hence when appropriate we follow a similar strategy for our proof for the Robust KronPCA case.

Consider the more general M-estimation problem
\begin{equation}
\label{Eq:General}
\min_{(\mathbf{\theta}_k)_{k\in I}} \mathcal{L}\left(\sum_{k\in I} \mathbf{\theta}_k\right) + \sum_{k\in I} \lambda_k \mathcal{R}_k(\mathbf{\theta}_k),
\end{equation}
where $\mathcal{L}$ is a convex differentiable loss function, the regularizers $\mathcal{R}_{k}$ are norms, with regularization parameters $\lambda_k \geq 2\mathcal{R}_k^*(\nabla_{\theta_k}\mathcal{L}(\mathbf{\theta}^*))$. $\mathcal{R}_k^*$ is the dual norm of the norm $\mathcal{R}_k$, $\nabla_\theta$ denotes the gradient with respect to $\theta$, and $\theta^*$ is the true parameter value. 

To emphasize that $\mathcal{L}$ depends on the observed training data $\mathbf{X}$ (in our case through the sample covariance), we also write $\mathcal{L}(\theta; \mathbf{X})$. Let $\mathcal{M}_k$ be the model subspace associated with the constraints enforced by $\mathcal{R}_k$ \cite{yang2013dirty}.
Assume the following conditions are satisfied:
\begin{enumerate}
\item \label{Cond:1} The loss function $\mathcal{L}$ is convex and differentiable.
\item \label{Cond:2} Each norm $\mathcal{R}_k$ ($k \in \mathcal{I}$) is decomposable with respect to the subspace pairs $(\mathcal{M}_k,\bar{\mathcal{M}}_k^\perp)$, where $\mathcal{M}_k\subseteq \bar{\mathcal{M}}_k$.
\item \label{Cond:3} (Restricted Strong Convexity) For all $\mathbf{\Delta}\in \Omega_k$, where $\Omega_k$ is the parameter space for parameter component $k$,
\begin{align}
\delta\mathcal{L}(\mathbf{\Delta}_k;\mathbf{\theta}^*) &:= \mathcal{L}(\mathbf{\theta}^* + \mathbf{\Delta}_k)-\mathcal{L}(\mathbf{\theta}^*) - \left\langle \mathbf{\Delta}_\theta\mathcal{L}(\mathbf{\theta}^*),\mathbf{\Delta}_k \right\rangle\nonumber\\
 &\geq \kappa_{\mathcal{L}}\|\mathbf{\Delta}_k\|^2 - g_k\mathcal{R}_k^2(\mathbf{\Delta}_k),
\end{align}
where $\kappa_\mathcal{L}$ is a ``curvature" parameter, and $g_k\mathcal{R}^2_k(\mathbf{\Delta}_k )$ is a ``tolerance" parameter.
\item \label{Cond:4}(Structural Incoherence) For all $\mathbf{\Delta}_k \in \Omega_k$,
\begin{align}
|\mathcal{L}&(\mathbf{\theta}^* + \sum_{k\in I} \mathbf{\Delta}_k) + (|I| - 1)\mathcal{L}(\mathbf{\theta}^*) - \sum_{k \in I }\mathcal{L}(\mathbf{\theta}^* + \mathbf{\Delta}_k)| \nonumber\\
&\leq \frac{\kappa_\mathcal{L}}{2}\sum_{k\in I} \|\mathbf{\Delta}_k\|^2 + \sum_{k\in I} h_k\mathcal{R}_k^2(\mathbf{\Delta}_k).
\end{align}
\end{enumerate}
Define the \textit{subspace compatibility constant} as $\Psi_k(\mathcal{M},\|\cdot\|):=\sup_{u\in \mathcal{M}\setminus\{0\}}\frac{\mathcal{R}_k(u)}{\|u\|}$. Given these assumptions, the following theorem holds (Corollary 1 in \cite{yang2013dirty}):
\begin{theorem}
\label{Thm:Gen}
Suppose that the subspace-pairs are chosen such that the true parameter values $\mathbf{\theta}_{k}^*\in\mathcal{M}_k$. Then the parameter error bounds are given as:
\begin{equation}
\|\hat{\mathbf{\theta}} - \mathbf{\theta}^*\| \leq \left(\frac{3|I|}{2\bar{\kappa}}\right)\max_{k\in I} \lambda_k \Psi_k(\bar{\mathcal{M}}_k).
\end{equation}
where
\begin{align*}
\bar{\kappa} := \frac{\kappa_{\mathcal{L}}}{2} - 32 \bar{g}|I|\left(\max_{k\in I} \lambda_k \Psi_k(\bar{\mathcal{M}}_k)\right)^2,\\
\bar{g} := \max_k \frac{1}{\lambda_k}\sqrt{g_k + h_k}.
\end{align*}
\end{theorem}

We can now prove Theorem \ref{Thm:HDC}.
\begin{proof}[Proof of Theorem \ref{Thm:HDC}]
To apply Theorem \ref{Thm:Gen} to the KronPCA estimation problem, we first check the conditions.
In our objective function \eqref{Eq:SparseOpt}, we have a loss $\mathcal{L}(\mathbf{\Sigma};\mathbf{X}) = \|\mathbf{\Sigma}-\hat{\mathbf{\Sigma}}_{SCM}\|_F^2$, which of course satisfies condition \ref{Cond:1}. It was shown in \cite{yang2013dirty} that the nuclear norm and the 1-norm both satisfy Condition \ref{Cond:2} with respect to $\mathcal{M}_\Theta, \bar{\mathcal{M}}_\Theta$ and $\mathcal{M}_\Gamma, \bar{\mathcal{M}}_\Gamma$ respectively. Hence we let the two $\mathcal{R}_k$ be the nuclear norm ($\mathcal{R}_\Theta(\cdot) = \|\cdot\|_*$) and the 1-norm ($\mathcal{R}_\Gamma(\cdot) = \|\cdot\|_1$) terms in \eqref{Eq:SparseOpt}. The restricted strong convexity condition (Condition \ref{Cond:3}) holds trivially with $\kappa_\mathcal{L} = 1$ and $g_k = 0$ \cite{yang2013dirty}.

It was shown in \cite{yang2013dirty} that for the linear Frobenius norm mismatch term ($\mathcal{L}(\mathbf{\Sigma}) = \|\mathbf{\Sigma}-\hat{\mathbf{\Sigma}}_{SCM}\|_F^2$) that we use in \eqref{Eq:SparseOpt}, the following simpler structural incoherence condition implies Condition \ref{Cond:4} with $h_k = 0$:
\begin{align}
\max&\left\{\sigma_{max}\left(\mathcal{P}_{\bar{\mathcal{M}}_\Theta}\mathcal{P}_{\bar{\mathcal{M}}_\Gamma} \right),\sigma_{max}\left(\mathcal{P}_{\bar{\mathcal{M}}_\Theta}\mathcal{P}_{\bar{\mathcal{M}}_\Gamma^\perp} \right),\right.\\\nonumber
&\left.\sigma_{max}\left(\mathcal{P}_{\bar{\mathcal{M}}_\Theta^\perp}\mathcal{P}_{\bar{\mathcal{M}}_\Gamma} \right),\sigma_{max}\left(\mathcal{P}_{\bar{\mathcal{M}}_\Theta^\perp}\mathcal{P}_{\bar{\mathcal{M}}_\Gamma^\perp} \right)\right\}\leq \frac{1}{16\Lambda^2}
\end{align}
where $\Lambda = \max_{k_1,k_2}\left\{2 + \frac{3\lambda_{k_1} \Psi_{k_1}(\bar{\mathcal{M}}_{k_1})}{\lambda_{k_2} \Psi_{k_2}(\bar{\mathcal{M}}_{k_2})}\right\}$.

The subspace compatibility constants are as follows \cite{yang2013dirty}:
\begin{align}
&\Psi_{\mathbf{\Theta}}(\bar{\mathcal{M}}_{\mathbf{\Theta}}) = \sup_{\mathbf{\Delta}\in\bar{\mathcal{M}}_{\mathbf{\Theta}}\setminus\{0\}}\frac{\|\mathbf{\Delta}\|_*}{\|\mathbf{\Delta}\|_F} \leq \sqrt{2r},\\\nonumber
&\Psi_{\mathbf{\Gamma}}(\bar{\mathcal{M}}_{\mathbf{\Gamma}}) = \sup_{\mathbf{\Delta}\in\bar{\mathcal{M}}_{\mathbf{\Gamma}}\setminus\{0\}}\frac{\|\mathbf{\Delta}\|_1}{\|\mathbf{\Delta}\|_F} \leq \sqrt{s},
\end{align}
where $r$ is the rank of $\mathbf{\Theta}$ and $s$ is the number of nonzero entries in $\mathbf{\Gamma}$. The first follows from the fact that for all $\mathbf{\Theta} \in \bar{\mathcal{M}}_{\Theta}$, $\mathrm{rank}(\mathbf{\Theta}) \leq 2r$ since both the row and column spaces of $\mathbf{\Theta}$ must be of rank $r$ \cite{yang2013dirty}. Hence, we have that
\begin{equation}
\Lambda = 2 + \max\left\{\frac{3\beta\sqrt{2r}}{\lambda\sqrt{s}},\frac{3\lambda\sqrt{s}}{\beta\sqrt{2r}}\right\}.
\end{equation}

Finally, we need to show that both of the regularization parameters satisfy $\lambda_k \geq 2\mathcal{R}_k^*(\nabla_{\theta_k}\mathcal{L}(\mathbf{\theta}^*;\mathbf{X}))$, i.e.
\begin{align}
\label{Eq:ParCond1}
\lambda_\Theta &\geq 2\mathcal{R}_\Theta^*(\nabla_{\Theta}\mathcal{L}(\mathbf{\Theta}_0 + \mathbf{\Gamma}_0;\mathbf{X}))\\\nonumber
\lambda_\Gamma &\geq 2\mathcal{R}_\Gamma(\nabla_{\Gamma}\mathcal{L}(\mathbf{\Theta}_0 + \mathbf{\Gamma}_0;\mathbf{X}))
\end{align}
with high probability. Since the 1-norm is invariant under rearrangement, the argument from \cite{yang2013dirty} still holds and we have that
\begin{equation}
\lambda_\Gamma = 32 \rho(\mathbf{\Sigma}_0)\sqrt{\frac{\log p_tp_s}{n}}
\end{equation}
satisfies \eqref{Eq:ParCond1} with probability at least $1-2\exp(-c_2 \log p_sp_t)$.

For the low rank portion, \eqref{Eq:ParCond1} will hold if \cite{yang2013dirty}
\begin{equation}
\label{Eq:LambdaCond1}
\lambda_\Theta \geq 4 \|\mathcal{R}(\hat{\mathbf{\Sigma}}_{SCM}-\mathbf{\Sigma}_0)\|.
\end{equation}

From \cite{tsiliArxiv} we have that for $t_0\geq f(\epsilon) = 4C \log (1+\frac{2}{\epsilon})$ ($C$ absolute constant given in \cite{tsiliArxiv}), $C$ an absolute constant, and $\alpha \geq 1$
\begin{align}
\|\mathcal{R}(\hat{\mathbf{\Sigma}}_{SCM}-\mathbf{\Sigma}_0)\| \leq  \frac{\|\mathbf{\Sigma}_0\|t_0}{1-2\epsilon}\frac{p_t^2 + p_s^2 + \log M}{n}
\end{align}
with probability at least $1-2M^{-t_0/4C}$ and otherwise
\begin{align}
\|\mathcal{R}(\hat{\mathbf{\Sigma}}_{SCM}-\mathbf{\Sigma}_0)\| \leq\frac{\|\mathbf{\Sigma}_0\|\sqrt{t_0}}{1-2\epsilon}\sqrt{\frac{p_t^2 + p_s^2 + \log M}{n}}
\end{align}
with probability at least $1-2M^{-t_0/4C}$. Thus our choice of $\lambda_\theta$ satisfies \eqref{Eq:LambdaCond1} with high probability. To satisfy the constraints on $t$, we need $t_0 > f^2(\epsilon)$. Clearly, $\epsilon$ can be adjusted to satisfy the constraint and
\begin{equation}
\label{Eq:k}
k = 4/(1-2\epsilon).
\end{equation}
Recalling the sparsity probability $1-2\exp(-c_2 \log p_tp_s)$, the union bound implies \eqref{Eq:ParCond1} is satisfied for both regularization parameters with probability at least $1-2\exp(-c_2 \log p_tp_s)-2\exp(-(t_0/4C)\log M)\geq 1-c\exp(-c_0 \log p_t p_s)$ and the proof of Theorem \ref{Thm:HDC} is complete.
\end{proof}

Next, we present the proof for Theorem \ref{Thm:HDCToep}, emphasizing only the parts that differ from the non Toeplitz proof of Theorem \ref{Thm:HDC}, since much of the proof for the previous theorem carries over to the Toeplitz case. Let 
\begin{align}
\mathbf{\Delta}_n &= \mathcal{R}(\mathbf{W})\\\nonumber 
\mathbf{W} &= \hat{\mathbf{\Sigma}}_{SCM} - \mathbf{\Sigma}_0.
\end{align}
We require the following corollary based on an extension of a theorem in \cite{tsiliArxiv} to the Toeplitz case:
\begin{corollary}
\label{Cor:ToepNorm}
Suppose $\mathbf{\Sigma}_0$ is a $p_tp_s\times p_tp_s$ covariance matrix, $\|\mathbf{\Sigma}_0\|_2$ is finite for all $p_t,p_s$, and let $M = \max(p_t,p_s,n)$. Let $\epsilon'<0.5$ be fixed and assume that $t_0 \geq f(\epsilon)$ and $C = \max(C_1,C_2)> 0$. We have that
\begin{equation}
\|\mathbf{\Delta}_n\|_2 \leq\frac{\|\mathbf{\Sigma}_0\|}{1-2\epsilon'}\max\left\{t_0\alpha^2,\sqrt{t_0}\alpha \right\}
\end{equation}
with probability at least $1-2M^{-\frac{t_0}{4C}}$, where
\begin{equation}
\alpha = \frac{2p_t + p_s^2 + \log M}{n}.
\end{equation}
\end{corollary}
The proof of this result is in Appendix \ref{App:B}.

\begin{proof}[Proof of Theorem \ref{Thm:HDCToep}]
Adjusting for the objective in \eqref{Eq:OptProbToep}, let the regularizers $\mathcal{R}_k$ be $\mathcal{R}_\Theta(\cdot) = \|\cdot\|_*$ and $\mathcal{R}_\Gamma(\mathbf{M}) = \sum_{j\in\mathcal{I}}c_j\|\mathbf{M}_{j+p_t}\|_1$. Condition \ref{Cond:1} still holds as in the general non-Toeplitz case, and Condition \ref{Cond:2} holds because $\mathcal{R}_\Gamma$ is a positively weighted sum of norms, forming a norm on the product space (which is clearly the entire space). $\mathcal{R}_\Gamma$ is decomposable because the 1-norm is decomposable and the overall model subspace is the product of the model subspaces for each row. The remaining two conditions trivially remain the same from the non Toeplitz case.

The subspace compatibility constant remains the same for the nuclear norm, and for the sparse case we have for all $\mathbf{\Delta}$
\begin{align}
\mathcal{R}_{\Gamma}(\mathbf{ \Delta}) \leq \|\mathbf{\Delta}\|_1
\end{align}
hence, the supremum under the 1-norm is greater than the supremum under the row weighted norm. Thus, the subspace compatibility constant is still less than or equal to $\sqrt{s}$, where $s$ is now the number of nonzero entries in $\mathbf{P}\mathcal{R}(\mathbf{\Gamma})$. A tighter bound is achievable if the degree of sparsity in each row is known.

We now show that the regularization parameters chosen satisfy \eqref{Eq:ParCond1} with high probability. For the sparse portion, we need to find the dual of $\mathcal{R}_\Gamma$, defined as
\begin{align}
\mathcal{R}^*_{\Gamma}(\mathbf{Z}) = \sup\left\{\langle\langle \mathbf{Z},\mathbf{X}\rangle\rangle | \mathcal{R}_{\Gamma}(\mathbf{X}) \leq 1\right\},
\end{align}
where $\langle\langle \mathbf{Z},\mathbf{X}\rangle\rangle = \mathrm{trace}\{\mathbf{Z}^T\mathbf{X}\}$.
Let the matrix $\mathbf{P}_1 = \mathrm{diag}\{\{\sqrt{p_t-|j|}\}_{j = -p_t+1}^{p_t-1}\}$. Define the matrices $\mathbf{X}'$ such that $\mathbf{X}' = \mathbf{P}_1^{-1}\mathbf{X}$. Then $\mathcal{R}_{\Gamma}(\mathbf{X}) = \|\mathbf{X}'\|_1$ and
\begin{align}
\mathcal{R}^*_{\Gamma}(\mathbf{Z}) &= \sup\left\{\langle\langle \mathbf{P}_1\mathbf{Z},\mathbf{X}'\rangle\rangle \:|\: \|\mathbf{X}'\|_1 \leq 1\right\}\\\nonumber
&=\|\mathbf{P}_1\mathbf{Z}\|_{\infty}
\end{align}
since the dual of the 1-norm is the $\infty$-norm. 
From \cite{agarwal2012noisy}, \eqref{Eq:ParCond1} now takes the form
\begin{align}
\lambda_\Gamma &\geq 4\|\mathbf{P}_1 \mathbf{PW}\|_\infty
= \|\tilde{\mathbf{W}}\|_\infty
\end{align}
where
\begin{align}
\label{Eq:CS}
\tilde{{W}}_{j+p_t,i} &= \sum_{\ell\in\mathcal{K}(j)} {W}_{\ell,i}.
\end{align}
Hence
\begin{align}
|\tilde{{W}}_{j+p_t,i}|&\leq (p_t-|j|) \|\mathbf{W}\|_\infty\\\nonumber
\|\tilde{\mathbf{W}}\|_\infty &\leq p_t\|\mathbf{W}\|_\infty.
\end{align}
From \cite{agarwal2012noisy} (via the union bound), we have
\begin{align}
\mathrm{Pr}\left(\|\mathbf{W}\|_\infty> 8\rho (\mathbf{\Sigma}) \sqrt{\frac{\log p_t p_s}{n}}\right) \leq 2 \exp(-c_2\log(p_tp_s)),
\end{align}
giving
\begin{equation}
\mathrm{Pr}\left(\|\tilde{\mathbf{W}}\|_\infty > 8\rho (\mathbf{\Sigma})p_t \sqrt{\frac{\log p_t p_s}{n}}\right) \leq 2\exp(-c_2\log(p_tp_s)),
\end{equation}
which demonstrates that our choice for $\lambda_\Gamma$ is satisfactory with high probability.


As before, for the low rank portion \eqref{Eq:ParCond1} will hold if \cite{yang2013dirty}
\begin{equation}
\label{Eq:LambdaCond}
\lambda_\Theta \geq 4 \|\mathbf{P}\mathcal{R}(\hat{\mathbf{\Sigma}}_{SCM}-\mathbf{\Sigma}_0)\|.
\end{equation}

From Corollary \ref{Cor:ToepNorm} we have that for $t\geq f(\epsilon)$, $C$ an absolute constant, and $\alpha \geq 1$
\begin{align}
\|\mathbf{P}\mathcal{R}(\hat{\mathbf{\Sigma}}_{SCM}-\mathbf{\Sigma}_0)\| \leq  \frac{\|\mathbf{\Sigma}_0\|t_0}{1-2\epsilon}\frac{2p_t + p_s^2 + \log M}{n}
\end{align}
with probability at least $1-2M^{-t_0/4C}$ and otherwise
\begin{align}
\|\mathbf{P}\mathcal{R}(\hat{\mathbf{\Sigma}}_{SCM}-\mathbf{\Sigma}_0)\| \leq\frac{\|\mathbf{\Sigma}_0\|\sqrt{t_0}}{1-2\epsilon}\sqrt{\frac{2p_t + p_t^2 + \log M}{n}}
\end{align}
with probability at least $1-2M^{-t_0/4C}$. Hence, in the same way as in the non-Toeplitz case we have with high probability
\begin{align}
\|\hat{\tilde{\mathbf{L}}} - &\tilde{\mathbf{L}}_0\|_F \leq\\\nonumber &6\max \left\{k\|\mathbf{\Sigma}_0\| \sqrt{r}\max(\alpha^2,\alpha), 32 \rho(\mathbf{\Sigma}_0)\sqrt{\frac{s\log p_tp_s}{n}} \right\}
\end{align}
and since $\mathbf{L} = \mathbf{P}^T\tilde{\mathbf{L}}$ the theorem follows.

\end{proof}
\section{Gaussian Chaos Operator Norm Bound}
\label{App:B}
We first note the following corollary from \cite{tsiliArxiv}:
\begin{corollary}
\label{Corr:B1}
Let $\mathbf{x}\in \mathcal{S}_{p_t^2-1}$ and $\mathbf{y}\in \mathcal{S}_{p_s^2-1}$. Let $\mathbf{z}_i \sim \mathcal{N}(0,\mathbf{\Sigma}_0)$, $i=1,\dots,n$ be $p_tp_s$ dimensional iid training samples. Let $\mathbf{\Delta}_n =\mathcal{R}(\hat{\mathbf{\Sigma}}_{SCM} - \mathbf{\Sigma}_0)= \mathcal{R}(\frac{1}{n}\sum_i \mathbf{z}_i\mathbf{z}_i^T - \mathbf{\Sigma}_0)$. Then for all $\tau>0$,
\begin{equation}
\mathrm{Pr}(|\mathbf{x}^T\mathbf{\Delta}_n \mathbf{y}| \geq \tau) \leq \exp\left(\frac{-n\tau^2/2}{C_1 \|\mathbf{\Sigma}_0\|_2^2 + C_2\|\mathbf{\Sigma}_0\|_2 \tau}\right)
\end{equation}
where $C_1,C_2$ are absolute constants.

\end{corollary}

The proof (appropriately modified from that of a similar theorem in \cite{tsiliArxiv}) of Corollary \ref{Cor:ToepNorm} then proceeds as follows:
\begin{proof}
Define $\mathcal{N}(\mathcal{S}_{d'-1},\epsilon')$ as an $\epsilon'$ net on $\mathcal{S}_{d'-1}$. Choose $\mathbf{x}_1\in \mathcal{S}_{2p_t-2},\mathbf{y}_1\in \mathcal{S}_{p_s^2-1}$ such that $|\mathbf{x}_1^T \mathbf{P} \mathbf{\Delta}_n \mathbf{y}_1| = \|\mathbf{P}\mathbf{\Delta}_n\|_2$. By definition, there exists $\mathbf{x}_2 \in \mathcal{N}(\mathcal{S}_{2p_t-2},\epsilon') ,\mathbf{y}_2 \in \mathcal{N}(\mathcal{S}_{p_s^2-1},\epsilon')$ such that $\|\mathbf{x}_1-\mathbf{x}_2\|_2 \leq \epsilon',\|\mathbf{y}_1-\mathbf{y}_2\|_2 \leq \epsilon'$. Then
\begin{align}
|\mathbf{x}_1^T\mathbf{P}\mathbf{\Delta}_n \mathbf{y}_1| - |\mathbf{x}_2^T\mathbf{P}\mathbf{\Delta}_n \mathbf{y}_2| &\leq |\mathbf{x}_1^T\mathbf{P}\mathbf{\Delta}_n \mathbf{y}_1 - \mathbf{x}_2^T\mathbf{P}\mathbf{\Delta}_n \mathbf{y}_2|\\\nonumber
 &\leq 2\epsilon'\|\mathbf{P}\mathbf{\Delta}_n\|_2.
\end{align}
We then have
\begin{align}
&\|\mathbf{P}\mathbf{\Delta}_n\|_2(1-2\epsilon')\\\nonumber
 &\leq \max\left\{|\mathbf{x}_2^T\mathbf{P}\mathbf{\Delta}_n \mathbf{y}_2| : \mathbf{x}_2 \in \mathcal{N}(\mathcal{S}_{2p_t-2},\epsilon') ,\right.\\\nonumber
 &\quad \quad \: \left.\mathbf{y}_2 \in \mathcal{N}(\mathcal{S}_{p_s^2-1},\epsilon'),\|\mathbf{x}_1-\mathbf{x}_2\|_2 \leq \epsilon',\|\mathbf{y}_1-\mathbf{y}_2\|_2 \leq \epsilon' \right\}\\\nonumber
&\leq \max\left\{|\mathbf{x}_2^T\mathbf{P}\mathbf{\Delta}_n \mathbf{y}_2| : \mathbf{x}_2 \in \mathcal{N}(\mathcal{S}_{2p_t-2},\epsilon'),\right.\\\nonumber
&\quad \quad \: \left. \mathbf{y}_2 \in \mathcal{N}(\mathcal{S}_{p_s^2-1},\epsilon')\right\}
\end{align}
since $|\mathbf{x}_1^T \mathbf{P} \mathbf{\Delta}_n \mathbf{y}_1| = \|\mathbf{P}\mathbf{\Delta}_n\|_2$. Hence
\begin{align}
&\|\mathbf{P}\mathbf{\Delta}_n\|_2\\\nonumber
&\leq \frac{1}{1-2\epsilon'}\max_{\mathbf{x} \in \mathcal{N}(\mathcal{S}_{2p_t-2},\epsilon') ,\mathbf{y} \in \mathcal{N}(\mathcal{S}_{p_s^2-1},\epsilon')}|\mathbf{x}^T\mathbf{P}\mathbf{\Delta}_n \mathbf{y}|.
\end{align}
From \cite{tsiliArxiv}
\begin{equation}
\mathrm{card}(\mathcal{N}(\mathcal{S}_{d'-1},\epsilon')) \leq \left(1+\frac{2}{\epsilon'}\right)^{d'}
\end{equation}
which allows us to use the union bound.
\begin{align}
&\mathrm{Pr}(\|\mathbf{P} \mathbf{\Delta}_n\|_2 > \epsilon')\\\nonumber
&\leq \mathrm{Pr}\left(\max_{\mathbf{x} \in \mathcal{N}(\mathcal{S}_{2p_t-2},\epsilon') ,\mathbf{y} \in \mathcal{N}(\mathcal{S}_{p_s^2-1},\epsilon')}|\mathbf{x}^T\mathbf{P}\mathbf{\Delta}_n \mathbf{y}|\geq \epsilon(1-2\epsilon')\right)\\\nonumber
&\leq \mathrm{Pr}\left(\bigcup_{\mathbf{x} \in \mathcal{N}(\mathcal{S}_{2p_t-2},\epsilon') ,\mathbf{y} \in \mathcal{N}(\mathcal{S}_{p_s^2-1},\epsilon')}|\mathbf{x}^T\mathbf{P}\mathbf{\Delta}_n \mathbf{y}|\geq \epsilon(1-2\epsilon')\right)\\\nonumber
&\leq \mathrm{card}(\mathcal{N}(\mathcal{S}_{2p_t-2},\epsilon'))\mathrm{card}(\mathcal{N}(\mathcal{S}_{p_s^2-1},\epsilon'))\\\nonumber
&\quad \times \max_{\mathbf{x} \in \mathcal{N}(\mathcal{S}_{2p_t-2},\epsilon') ,\mathbf{y} \in \mathcal{N}(\mathcal{S}_{p_s^2-1},\epsilon')}\mathrm{Pr}(|\mathbf{x}^T\mathbf{P}\mathbf{\Delta}_n \mathbf{y}|\geq \epsilon(1-2\epsilon'))\\\nonumber
&\leq \left(1+\frac{2}{\epsilon'}\right)^{2p_t+p_s^2}\\\nonumber
&\quad \times \max_{\mathbf{x} \in \mathcal{N}(\mathcal{S}_{2p_t-2},\epsilon') ,\mathbf{y} \in \mathcal{N}(\mathcal{S}_{p_s^2-1},\epsilon')}\mathrm{Pr}(|\mathbf{x}^T\mathbf{P}\mathbf{\Delta}_n \mathbf{y}|\geq \epsilon(1-2\epsilon')).
\end{align}
Note that
\begin{align}
\|\mathbf{x}^T\mathbf{P}\|_2^2 =& \sum_j \frac{x_{j+p_t}^2}{p_t-|j|}(p_t-|j|)\\\nonumber
=& \sum_jx_j^2 = \|\mathbf{x}\|_2^2 = 1,
\end{align}
so $\mathbf{x}^T\mathbf{P} \in \mathcal{S}_{p_t^2-1}$. We can thus use Corollary \ref{Corr:B1}, giving
\begin{align}
&\mathrm{Pr}(\|\mathbf{P} \mathbf{\Delta}_n\|_2 > \epsilon')\\\nonumber
&\leq 2\left(1+\frac{2}{\epsilon'}\right)^{2p_t+p_s^2}\exp\left(\frac{-n\epsilon^2(1-2\epsilon')^2/2}{C_1\|\mathbf{\Sigma}_0\|_2^2 + C_2\|\mathbf{\Sigma}_0\|_2\epsilon(1-2\epsilon')}\right).
\end{align}
Two regimes emerge from this expression. The first is where $\epsilon \leq \frac{C_1\|\mathbf{\Sigma}_0\|_2}{C_2(1-2\epsilon')}$, which allows
\begin{align}
\mathrm{Pr}&(\|\mathbf{P} \mathbf{\Delta}_n\|_2 > \epsilon)\\\nonumber
&\leq 2\left(1+\frac{2}{\epsilon'}\right)^{2p_t+p_s^2}\exp\left(\frac{-n\epsilon^2(1-2\epsilon')^2/2}{2C_1\|\mathbf{\Sigma}_0\|_2^2}\right).
\end{align}
Choose
\begin{align}
\epsilon = \frac{\sqrt{t_0}\|\mathbf{\Sigma}_0\|_2}{1-2\epsilon'}\sqrt{\frac{2p_t + p_s^2 + \log M}{n}}.
\end{align}
This gives:
\begin{align}
\mathrm{Pr}&\left(\|\mathbf{P} \mathbf{\Delta}_n\|_2 > \frac{\sqrt{t_0}\|\mathbf{\Sigma}_0\|_2}{1-2\epsilon'}\sqrt{\frac{2p_t + p_s^2 + \log M}{n}}\right)\\\nonumber
&\leq 2\left(1+\frac{2}{\epsilon'}\right)^{2p_t+p_s^2}\exp\left(\frac{-t^2 (2p_t + p_s^2 + \log M)}{4C_1}\right)\\\nonumber
&\leq 2\left(\left(1+\frac{2}{\epsilon'}\right)e^{-\frac{t_0}{4C_1}}\right)^{2p_t+p_s^2}M^{-t_0/(4C_1)}\\\nonumber
&\leq 2M^{-t_0/(4C_1)}.
\end{align}
The second regime ($\epsilon > \frac{C_1\|\mathbf{\Sigma}_0\|_2}{C_2(1-2\epsilon')}$) allows us to set $\epsilon$ to
\begin{align}
\epsilon = \frac{t_0\|\mathbf{\Sigma}_0\|_2}{1-2\epsilon'}{\frac{2p_t + p_s^2 + \log M}{n}}
\end{align}
which gives
\begin{align}
\mathrm{Pr}&\left(\|\mathbf{P} \mathbf{\Delta}_n\|_2 > \frac{t_0\|\mathbf{\Sigma}_0\|_2}{1-2\epsilon'}{\frac{2p_t + p_s^2 + \log M}{n}}\right)\\\nonumber
&\leq 2\left(1+\frac{2}{\epsilon'}\right)^{2p_t+p_s^2}\exp\left(\frac{-t (2p_t + p_s^2 + \log M)}{4C_2}\right)\\\nonumber
&\leq 2M^{-t_0/(4C_2)}.
\end{align}
Combining both regimes (noting that $t_0>1$ and $\sqrt{t_0} C_1/C_2>1$) completes the proof.
\end{proof}

%
%
%
%
%
%
%
%
%
%
%
%

\end{appendices}

\bibliographystyle{ReferenceFormat}
\bibliography{CAMSAP_bib}

\begin{thebibliography}{10}

\bibitem{greenewaldArxiv}
K.~Greenewald, T.~Tsiligkaridis, and A.~Hero, ``Kronecker sum decompositions of
  space-time data,'' in {\em Proceedings of IEEE CAMSAP},  2013.

\bibitem{greenewaldSSP2014}
K.~Greenewald and A.~Hero, ``Regularized block toeplitz covariance matrix
  estimation via kronecker product expansions,'' in {\em Proceedings of IEEE
  SSP},  2014.

\bibitem{tsiliArxiv}
T.~Tsiligkaridis and A.~Hero, ``Covariance estimation in high dimensions via
  kronecker product expansions,'' {\em IEEE Trans. on Sig. Proc.}~{\bf 61}(21),
  pp.~5347--5360, 2013.

\bibitem{dutilleulFLIPFLOP}
P.~Dutilleul, ``The mle algorithm for the matrix normal distribution,'' {\em
  Journal of statistical computation and simulation}~{\bf 64}(2), pp.~105--123,
  1999.

\bibitem{dawid1981some}
A.~P. Dawid, ``Some matrix-variate distribution theory: notational
  considerations and a bayesian application,'' {\em Biometrika}~{\bf 68}(1),
  pp.~265--274, 1981.

\bibitem{tsiligkaridis2013convergence}
T.~Tsiligkaridis, A.~Hero, and S.~Zhou, ``On convergence of kronecker graphical
  lasso algorithms,'' {\em IEEE Trans. Signal Proc.}~{\bf 61}(7),
  pp.~1743--1755, 2013.

\bibitem{greenewaldSPIE2014}
K.~Greenewald and A.~Hero, ``Kronecker pca based spatio-temporal modeling of
  video for dismount classification,'' in {\em Proceedings of SPIE},  2014.

\bibitem{loan1992approximation}
C.~V. Loan and N.~Pitsianis, ``Approximation with kronecker products,'' in {\em
  Linear Algebra for Large Scale and Real Time Applications},  pp.~293--314,
  Kluwer Publications, 1993.

\bibitem{werner2008estimation}
K.~Werner, M.~Jansson, and P.~Stoica, ``On estimation of cov. matrices with
  kronecker product structure,'' {\em IEEE Trans. on Sig. Proc.}~{\bf 56}(2),
  pp.~478--491, 2008.

\bibitem{chandrasekaran2009}
V.~Chandrasekaran, S.~Sanghavi, P.~Parrilo, and A.~Willsky, ``Sparse and
  low-rank matrix decompositions,'' in {\em Communication, Control, and
  Computing, 2009. Allerton 2009. 47th Annual Allerton Conference on},
  pp.~962--967, Sept~2009.

\bibitem{chandrasekaran2010latent}
V.~Chandrasekaran, P.~A. Parrilo, and A.~S. Willsky, ``Latent variable
  graphical model selection via convex optimization,'' in {\em Communication,
  Control, and Computing (Allerton), 2010 48th Annual Allerton Conference on},
  pp.~1610--1613, IEEE, 2010.

\bibitem{candes2011robust}
E.~J. Cand{\`e}s, X.~Li, Y.~Ma, and J.~Wright, ``Robust principal component
  analysis?,'' {\em Journal of the ACM (JACM)}~{\bf 58}(3), p.~11, 2011.

\bibitem{yang2013dirty}
E.~Yang and P.~Ravikumar, ``Dirty statistical models,'' in {\em Advances in
  Neural Information Processing Systems},  pp.~611--619, 2013.

\bibitem{5540138}
Y.~Peng, A.~Ganesh, J.~Wright, W.~Xu, and Y.~Ma, ``Rasl: Robust alignment by
  sparse and low-rank decomposition for linearly correlated images,'' in {\em
  Computer Vision and Pattern Recognition (CVPR), 2010 IEEE Conference on},
  pp.~763--770, June~2010.

\bibitem{otazo2014low}
R.~Otazo, E.~Cand{\`e}s, and D.~K. Sodickson, ``Low-rank plus sparse matrix
  decomposition for accelerated dynamic mri with separation of background and
  dynamic components,'' {\em Magnetic Resonance in Medicine} , 2014.

\bibitem{mazumder2010spectral}
R.~Mazumder, T.~Hastie, and R.~Tibshirani, ``Spectral regularization algorithms
  for learning large incomplete matrices,'' {\em Journal of Machine Learning
  Research}~{\bf 11}, pp.~2287--2322, 2010.

\bibitem{mooreSSP2014}
B.~Moore, R.~Nadakutiti, and J.~Fessler, ``Improved robust pca using low-rank
  denoising with optimal singular value shrinkage,'' in {\em Proceedings of
  IEEE SSP},  2014.

\bibitem{kamm2000optimal}
J.~Kamm and J.~Nagy, ``Opt. kronecker product approx. of block toeplitz
  matrices,'' {\em SIAM Journal on Matrix Analysis and App.}~{\bf 22}(1),
  pp.~155--172, 2000.

\bibitem{pitsianis1997kronecker}
N.~P. Pitsianis, {\em The Kronecker product in approximation and fast transform
  generation}.
\newblock PhD thesis, Cornell University, 1997.

\bibitem{deckard2013design}
A.~Deckard, R.~C. Anafi, J.~B. Hogenesch, S.~B. Haase, and J.~Harer, ``Design
  and analysis of large-scale biological rhythm studies: a comparison of
  algorithms for detecting periodic signals in biological data,'' {\em
  Bioinformatics}~{\bf 29}(24), pp.~3174--3180, 2013.

\bibitem{agarwal2012noisy}
A.~Agarwal, S.~Negahban, M.~J. Wainwright, {\em et~al.}, ``Noisy matrix
  decomposition via convex relaxation: Optimal rates in high dimensions,'' {\em
  The Annals of Statistics}~{\bf 40}(2), pp.~1171--1197, 2012.

\end{thebibliography}

\end{document}